\documentclass{article}
\usepackage[utf8]{inputenc}
\usepackage[table]{xcolor}
\usepackage{soul}
\usepackage{titlesec}
\usepackage{setspace}
\usepackage{authblk}
\usepackage{url}
\usepackage{tikz}
\usepackage{tikzscale}
\usepackage{hyperref} 
\usepackage{pgfplots}
\usepackage{caption}
\usepackage{algorithm}
\usepackage{algorithmic}
\usepackage{amsmath}
\usepackage{amsthm}
\usepackage{bbm}
\usepackage{amsfonts}
\usepackage{amssymb}
\usepackage{tabularx}
\usepackage{tabulary}
\usepackage{multirow}
\usepackage{scalefnt}
\pgfplotsset{width=10cm,compat=1.15}
\usepackage{enumitem}
\usepackage{subcaption}
\setlist{nolistsep}
\usepackage{etoolbox}

\newtheoremstyle{mystyle}  
  {1pt}                   
  {1pt}                   
  {\itshape}              
  {}                      
  {\bfseries}             
  {.}                     
  {.5em}                  
  {}                      

\theoremstyle{mystyle}

\newtheorem{theo}{Theorem}
\newtheorem{coro}{Corollary}

\titlespacing\section{0pt}{5 pt}{5pt}
\titlespacing\subsection{0pt}{5 pt}{5pt}
\titlespacing\subsubsection{0pt}{5 pt}{5pt}

\usepackage{amsmath}
\usepackage{amsthm}
\usepackage{amssymb}
\usepackage{amsfonts}
\usepackage{amsxtra}     
\usepackage{verbatim}
\usepackage{mathrsfs}
\usepackage{accents}

\usepackage{algorithm}
\usepackage{algorithmic}

\makeatletter
\newcommand{\ALOOP}[1]{\ALC@it\algorithmicloop\ #1%
	\begin{ALC@loop}}
	\newcommand{\ENDALOOP}{\end{ALC@loop}\ALC@it\algorithmicendloop}

\makeatother

\usepackage{nomencl}
\makenomenclature
\usepackage{array}

\newcolumntype{F}[1]{%
    >{\raggedright\arraybackslash\hspace{0pt}}p{#1}}%
\newcolumntype{T}[1]{%
    >{\centering\arraybackslash\hspace{0pt}}p{#1}}%

\allowdisplaybreaks

\usepackage{natbib}
\setcitestyle{authoryear}
\title{Transmission Benefits and Cost Allocation under Ambiguity}
\author[1]{Han Shu\thanks{hs2226@cornell.edu}}
\author[2]{Jacob Mays\thanks{jacobmays@cornell.edu}}

\affil[1]{Systems Engineering, Cornell University, Ithaca, NY 14853}
\affil[2]{Civil and Environmental Engineering, Cornell University, Ithaca, NY 14853}

\begin{document}
\maketitle
\allowdisplaybreaks



\begin{abstract}
Disputes over cost allocation can present a significant barrier to investment in shared infrastructure. While it may be desirable to allocate cost in a way that corresponds to expected benefits, investments in long-lived projects are made under conditions of substantial uncertainty. In the context of electricity transmission, uncertainty combined with the inherent complexity of power systems analysis prevents the calculation of an estimated distribution of benefits that is agreeable to all participants. To analyze aspects of the cost allocation problem, we construct a model for transmission and generation expansion planning under uncertainty, enabling the identification of transmission investments as well as the calculation of benefits for users of the network. Numerical tests confirm the potential for realized benefits at the participant level to differ significantly from ex ante estimates. Based on the model and numerical tests we discuss several issues, including 1) establishing a valid counterfactual against which to measure benefits, 2) allocating cost to new and incumbent generators vs. solely allocating to loads, 3) calculating benefits at the portfolio vs. the individual project level, 4) identifying losers in a surplus-enhancing transmission expansion, and 5) quantifying the divergence between cost allocation decisions made ex ante and benefits realized ex post. 

\end{abstract}
\textbf{Keywords:} Electricity markets, transmission planning, cost allocation, uncertainty



\doublespacing

\section{Introduction}
A wealth of recent research finds that large-scale expansion of regional and interregional transmission infrastructure in the U.S. would bring economic, reliability, and environmental benefits~\citep{MacDonald2016,Seams2020,Brown2021}. Planned additions to transmission, however, currently fall well short of the level deemed beneficial in models, with studies of deeply decarbonized U.S. systems projecting a need to double or even triple transfer capacity in the coming decades~\citep{denholm2022examining,jenkins2021mission}. One of the major challenges holding up investment is cost allocation: since many stakeholders are likely to benefit from expanded transmission infrastructure, it is difficult to come to a consensus on how to divide project costs among them. While this issue is common to many types of shared infrastructure projects (see, e.g.,~\cite{Hamilton2021}), it may be particularly acute in the case of meshed electricity networks due to the potential for a change in one element to affect power flows across the entire system.

An underlying principle for cost allocation, codified by the Federal Energy Regulatory Commission (FERC) in Order~1000~\citep{FERC1000}, is that transmission costs should be allocated ``in a manner that is at least roughly commensurate with estimated benefits.'' In principle, sufficiently detailed planning models could be used both to establish the net social benefits of a project and to estimate a distribution of those benefits among users, and~\cite{Hogan2018} argues that these models would be the best available basis for determining a reasonable cost allocation. At the same time, both our information and our models fall well short of what would be needed for a precise computation, leading others to question the approach~\citep{BushnellWolak2017}. Judge Richard Cudahy of the U.S. Court of Appeals, dissenting in a case connected to Order~1000, articulates the opposing view as follows: ``The majority has expressed a need for more precise numbers about benefits, burdens and a variety of other aspects. Now it has enhanced that need by suggesting the use of cost-benefit analysis (a method, some think, of dressing up dubious numbers to reach more impressive solutions). I will say preliminarily that I think the majority is under the impression that somehow there is a mathematical solution to this problem, and I think that this is a complete illusion. Despite the frequency with which cost-benefit analysis is used, it does not resolve the difficulty of accurately or meaningfully measuring the costs and benefits involved with these grid strengthening projects. Cost allocation, particularly at these extraordinarily high voltages, is far from a precise science, and there are no mathematical solutions to determining benefits region by region or subregion by subregion''~\citep{FERCsuit2}. This skepticism has perhaps been validated by the emergence since Order~1000 of a wide range of cost allocation methods that have all been determined to meet the ``roughly commensurate'' threshold, leading to inconsistent treatment of similar projects based on the process by which they were approved. Ongoing disputes led FERC to reopen the issue of cost allocation in a 2022 Notice of Proposed Rulemaking~\citep{FERC_TxNOPR}.

This paper addresses several questions connected to the use of planning models in identifying beneficiaries of transmission expansion, calculating their benefits, and allocating cost in a manner consistent with the ``beneficiaries pay'' standard. Despite the challenges,~\cite{Adamson2018} reflects that ``In the absence of an economically and politically acceptable formula, a direct benefits modeling approach as advocated by Hogan may prove the only workable solution.'' The Organization of MISO States (OMS), which includes regulators from a politically diverse set of states in the Midcontinent Independent System Operator (MISO) region, largely endorses the approach in its Statement of Principles for cost allocation~\citep{OMS-principles}. In this context, the paper considers both how existing methods can create conflict by ignoring modeled benefits as well as the factors that may prevent a ``direct benefits modeling'' approach from being workable.

At least three categories of issues affect the computation of benefits and the translation of modeled benefits to a mutually agreeable cost allocation. The first category relates to the models themselves. Transmission planners use a variety of software tools to inform planning, which can be broadly split into 1) security analysis tools, involving detailed power flow models but no economic criteria, 2) production cost tools, which simulate market outcomes with a fixed resource mix, 3) expansion planning tools, which optimize the addition of new generation and transmission resources, and 4) resource adequacy tools, which evaluate the reliability provided by a given resource mix given simulated weather and outage scenarios. In principle, socially optimal transmission decisions could be given by an expansion planning model (number 3 above) that incorporated production cost simulations (2) on a number of scenarios comparable to that used in resource adequacy analysis (4) and including fully specified power flow constraints (1), while also considering the strategic behavior of market participants. The intractability of such a model means that benefit calculations must be performed on separate tools that are simplified along different dimensions, leading to the potential for benefits to be either omitted or duplicated. The second category relates to explicit disagreements on the value of particular benefits. While economic outcomes have a shared measure, benefits related to reliability and public policy are more difficult to quantify. In particular, different jurisdictions in the same regional market often assign different values to carbon reduction and air pollution mitigation. Additionally, some jurisdictions within a market may give more weight to producers of electricity (e.g., to support job creation), while others prioritize minimizing cost to consumers. The third category relates to uncertainty in the input parameters required for planning models (e.g., future demand growth and technology improvements). Even if the benefits could be quantified in a straightforward way, the significant uncertainty inherent to the system means that the ex ante estimates of expected benefits could be very different from the actual benefits seen ex post. Since participants are unlikely to agree on the probability of potential future scenarios, and may even benefit from strategically misrepresenting their views on those probabilities, they are unlikely to agree on estimated benefits.

We primarily address the first and third of these issues, leaving a more comprehensive discussion of the second to future work. Existing model-based approaches for cost allocation can be divided between those assessing benefits based on changes in power flows~\citep{galiana2003transmission,abou2009transmission,yang2015structural,avar2022new} or in prices~\citep{roustaei2014transmission,banez2017beneficiaries,banez2017estimating,Kristiansen2018,Hogan2018,ONeill2020}. While physical approaches are sometimes used in practice, and it is possible that physical usage correlates with economic benefits, the connection is not clear and we take the more direct economic approach. The economics-based models can be further divided between those computing benefits directly~\citep{roustaei2014transmission,Hogan2018,ONeill2020} and those employing concepts from cooperative game theory to address the bargaining power of different participants~\citep{banez2017beneficiaries,banez2017estimating,Kristiansen2018}. Conceding the salience of bargaining power, we pursue the former approach due to its clearer connection to the ``beneficiaries pay'' principle. Among the models analyzed, none explicitly include uncertainty and only~\cite{Kristiansen2018} includes recourse decisions in the form of generation investment. Along these lines, we extend the approach sketched on simple examples in~\cite{ONeill2020} to a stochastic program co-optimizing the expansion of transmission and generation over a long time horizon. While such models have been considered by many researchers~\citep{wogrin2021welfare,Munoz2014,Munoz2015,newlun2021adaptive}, the primary focus in the literature has been identifying high-quality planning solutions rather than investigating the implications for cost allocation. Using our model, we establish beneficiaries and calculate benefits under many possible realizations of uncertainty, providing a more comprehensive understanding of the implications of network expansion for all involved parties. While we do not explicitly model the effect that cost allocation decisions may have on the network expansion decisions themselves, as in~\cite{Bravo2016}, these potential consequences are a theme throughout the discussion.

Through theoretical analysis and a numerical study on a stylized version of the Electric Reliability Council of Texas (ERCOT) system, we discuss five issues:
\begin{enumerate}
    \item How to construct a valid counterfactual against which to measure benefits of a transmission investment. Here, our primary argument is that at a minimum models must include the different generation investments likely to arise in response to different transmission expansion decisions.
    \item When cost should be allocated to generators. While current practice in U.S. systems typically allocates cost to new interconnecting generators but then excludes them from subsequent cost allocation, we conclude that a direct benefits modeling approach would instead allow new generators to connect without cost but then allocate cost to them throughout their life.
    \item Whether to allocate costs on a project-by-project basis or as a portfolio. In the numerical study, allocation at the project level implies that positive cost is allocated to participants with negative net benefits overall; on this basis, we find that portfolio-based allocation is more consistent with the ``beneficiaries pay'' principle.
    \item The potential to compensate market participants who see negative expected benefits from expansion decisions. Here we suggest that the surplus gained from transmission expansion could in principle be used to compensate participants who see negative net benefits, potentially reducing conflicts. 
    \item The potential that participant-level benefits realized ex post will be significantly out of alignment with ex ante estimates. Again with the intent of reducing conflicts, we suggest the possibility of defining financial contracts that would effectively reallocate cost ex post to market participants based on realized benefits.
\end{enumerate}
While the first three are topics of active debate among regulators and thus have near-term policy implications, the last two raise issues for longer-term consideration.

\section{Stochastic Expansion Planning} \label{se:multistage}
As a basis for analyzing the cost allocation problem, we construct a two-stage stochastic program optimizing expansion of generation and transmission over an extended horizon given an agreed-upon set of scenarios. Capacity expansion can be posed either as a social planning problem~\citep{de2008transmission,newlun2021adaptive} or as a multi-agent game in a competitive market setting~\citep{Sauma2006,wogrin2021welfare}. Given the complexity of modeling strategic behavior, system operators at present rely on more straightforward optimization formulations~\citep{Lau2021}. We adopt the same approach, noting that because expansion of transmission tends to weaken the ability of generators to exercise local market power~\citep{Wolak2020}, inclusion of strategic considerations in our model would likely shift our estimates of the distribution of benefits away from generators toward consumers. The first stage of the stochastic program includes decisions for the present year, while the second stage includes decisions to be made in several subsequent years. While the analysis could also be extended to a multistage setting with each year corresponding to a stage, we use a two-stage approximation to ensure scalability in the numerical study.  


The problem is formulated as a mixed-integer programming (MIP) model with transmission line investment decisions as binary variables and generation investment decisions as continuous. Binary variables are needed to represent a key feature of transmission investments, namely, significant economies of scale. Further, it is typically impossible to build a transmission facility with a rating that exactly matches the need, as equipment is available only in a limited number of standardized voltage and power ratings. Transmission investments in the model can be selected from defined levels of expansion with costs reflecting economies of scale. For generation investments, we assume perfect competition and linear costs. These assumptions ensure that, conditional on the transmission network decisions, nodal electricity prices support a resource mix that maximizes long-term social welfare.

Rather than the development of the planning model itself, the primary focus of this study is the translation of the planning model results to cost allocation determinations. Many debates in transmission planning concern the selection of scenarios and benefit\textendash cost thresholds used to justify the investment, as well as the subjective valuation of non-quantified benefits~\citep{Hogan2018}. We set aside these issues, as it is sufficient for the discussion to have a planning tool that recommends transmission investments with positive expected net benefits in sample.  We assume that a stakeholder process is able to construct scenarios and associated probabilities for use in the model, but do not assume that these scenarios are exhaustive, that the selected probabilities are accurate, or that the chosen scenarios and probabilities match the beliefs of individual market participants. 

\subsection{Notation}
\textit{Sets:}
\begin{itemize}
\item[] $y/\mathcal{Y}$: time index (years)
\item[] $n/\mathcal{N}$: nodes in a scenario tree
\item[] $b/\mathcal{B}$ $(\mathcal{B}')$: buses (without reference bus)
\item[] $t/\mathcal{T}$: time blocks
\item[] $l/\mathcal{L}$: lines
\item[] $g/\mathcal{G}$: all generators
\item[] $g/\mathcal{G}_R \subseteq \mathcal{G}$: renewable generators
\item[] $g/\mathcal{G}_T \subseteq \mathcal{G}$: thermal generators
\item[] $q/\mathcal{Q}$: transmission capacity increment options
\item[] $i/\mathcal{I}$: power balance penalty curve segments
\end{itemize} \bigskip

\textit{Parameters:}
\begin{itemize}
\item[] $\zeta_{\delta(n)}$ ($\zeta_{y}$): discount factor of node $n$ (in time index $y$)
\item[] $C^{\text{INV}}_{n,g}$: annualized generation investment cost of generation technology $g$ in node $n$ per unit capacity (\$/MW)
\item[] $C^{\text{INV}}_{l,q}$: annualized transmission investment cost of line $l$ for expansion type $q$ (\$)
\item[] $C^{\text{FIX}}_g$: per unit fixed operation and maintenance cost of generation technology $g$ (\$/MW-yr)
\item[] $C^{\text{VOM}}_g$: per unit variable operation and maintenance cost of generation technology $g$ (\$/MWh)
\item[] $C^{\text{EN}}_{n,g}$: per unit production fuel cost of generation technology $g$ in node $n$ (\$/MWh)
\item[] $\gamma_i^{\text{PB}}$: penalty value of power balance violation in segment $i$ (\$/MWh)
\item[] $\gamma^{\text{LINE}}$: penalty value of transmission line violation (\$/MWh)
\item[] $\gamma^{\text{LOAD}}$: per unit benefit for serving load (\$/MWh)
\item[] $T_t$: duration of time block $t$ (h)
\item[] $CA_{b,g,t}$: capacity availability of generation technology $g$ located at bus $b$ at time $t$ 
\item[] $RPS_{n}$: renewable portfolio standard in node $n$ (\%)
\item[] $D_{n,b,t}$: demand at bus $b$ in time block $t$ in node $n$
\item[] $\Delta L_{q}$: transmission capacity increment $q$
\item[] $SF_{l,b}$: shift factor matrix indexed by $l \in \mathcal{L}, b \in \mathcal{B}'$
\item[] $\phi_n$: the probability of node $n$
\item[] $\overline{Z}_{i}$: maximum MW violation of power balance constraint for segment $i$
\end{itemize} \bigskip

\textit{Variables:}
\begin{itemize}
\item[] $c_{n}^{\text{cap}}$: the capital cost in node $n$ (\$)
\item[] $c_{n}^{\text{op}}$: the operation cost in node $n$ (\$)
\item[] $G_{n,b,g}$: total cumulative generation capacity of generation $g$ at bus $b$ in node $n$ (MW)
\item[] $\Delta G_{n,b,g}$: generation investment in generation technology $g$ at bus $b$ in node $n$ (MW)
\item[] $\Delta \overline{G}_{n,b,g}$: generation retirement of existing generation $g$ at bus $b$ in node $n$ (MW)
\item[] $L_{n,l}$: total cumulative transmission capacity of line $l$ in node $n$ (MW)
\item[] $w_{n,l,q}$: binary variable to decide transmission increment $q$ in line $l$ in node $n$ (MW)
\item[] $p_{n,b,g,t}$: generation dispatch of generation technology $g$ at bus $b$ at time $t$ in node $n$ (MW)
\item[] $z_{n,b,t,i}$: load curtailment segment $i$ at bus $b$ at time $t$ in node $n$ (MW)
\item[] $NI_{n,b,t}$: power net injection at bus $b$ in node $n$ at time $t$ (MW)
\item[] $sl_{n,l,t}$: slack variable for power flow on line $l$ in node $n$ at time $t$ (MW)
\end{itemize} \bigskip

\textit{Dual Variables:}
\begin{itemize}
\item[] $\pi_{n,b,t}$: locational marginal price (LMP) at time block $t$ at bus $b$ in node $n$ (\$/MWh)
\item[] $\theta_{n,b,g,t}$: marginal value of a unit of generation technology $g \in \mathcal{G}$ at time $t$ (\$/MWh)
\item[] $\nu_{n}$: unit price for contributing to the renewable portfolio standard in node $n$ (\$/MWh)
\end{itemize} \bigskip

\textit{Outputs:}
\begin{itemize}
\item[] $\mathcal{U}_b^{load}$: the aggregated load surplus at bus $b$
\item[] $\mathcal{U}_{b,g}^{gen}$: the per unit generation surplus of generation technology $g$ at bus $b$
\item[] $r_b^{load}$: cost allocation ratio of the transmission expansion to bus $b$ (\%)
\item[] $r_{b,g}^{gen}$: cost allocation ratio of the transmission expansion to existing generation $g$ at bus $b$ (\%)
\end{itemize}

\subsection{Formulation}

\begin{figure}[!h]
    \centering
    \includegraphics[width=7cm]{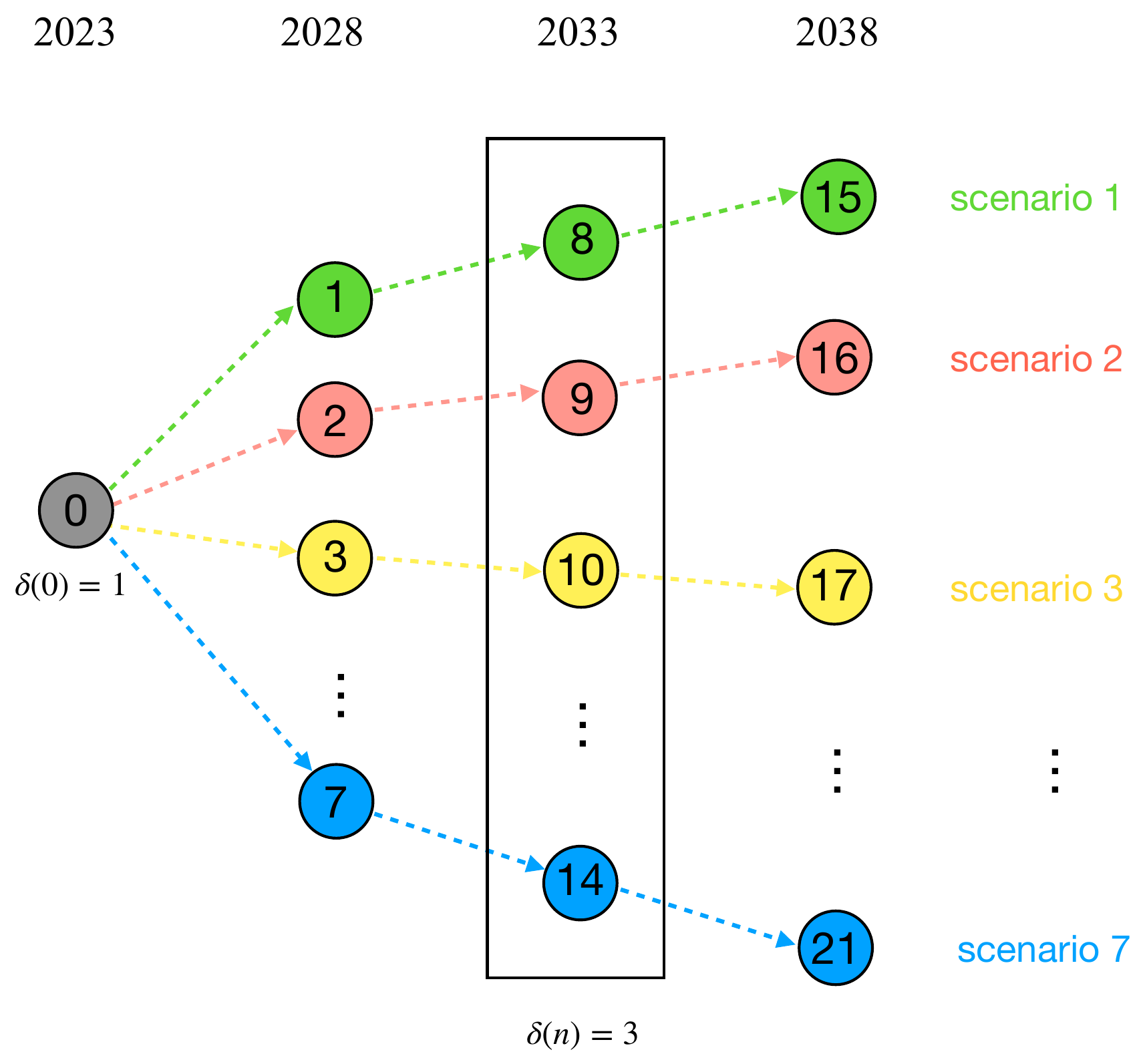}
    \caption{An illustration of a scenario tree with 7 scenarios and $\mathcal{Y} = \{1,2,3,4\}$.}
    \label{fig:scenario}
\end{figure}

We employ a scenario tree with nodes $n \in \mathcal{N}$ to represent the investment trajectory for the two-stage stochastic program. Each node represents a possible state of the world, associated with a set of data. The root node $n=0$ in the first stage represents the current state of the world. The unique predecessor of any node $n \neq 0$ is denoted as $n_{-}$ and the set of predecessors of node $n$ on the path from $n$ to the root node is denoted as $\mathcal{P}(n)$. The depth $\delta (n)$ of node $n$ is the number of nodes on the path to node $0$, with $\delta (0) = 1$. The depth $\delta (n)$ of node $n$ also corresponds to a time index $y \in \mathcal{Y}$. We use $\phi_n$ to represent the probability that the path taken through the scenario tree includes node $n$, with $\sum_{n \in \mathcal{N}: \delta (n) = y} \phi_n = 1 \text{ } \forall y \in \mathcal{Y}$. A visual representation of such a scenario tree with $\mathcal{Y} = \{1,2,3,4\}$ and $7$ scenarios is drawn in Figure~\ref{fig:scenario}. As indicated by the dashed lines, nodes at depth 2 and 3 have a unique successor, reflecting the two-stage simplification previously mentioned. This tree structure mimics the scenario-based planning performed by many system operators, but forces convergence to a single decision in the present year. The focus of the cost allocation discussion will be on transmission investments made in the present year.



In each node $n$, the capital cost includes the transmission and generation investment costs incurred due to the cumulative investment decisions made on the path from node $n$ to the root node, given by 
\begin{align}  \label{eq:capital_cost_node_n}
    c^{\text{cap}}_{n} = & \sum_{n' \in \mathcal{P}(n)} \sum_{l \in \mathcal{L}} \sum_{q \in \mathcal{Q}} C^{\text{INV}}_{l,q} w_{n',l,q} + \sum_{n' \in \mathcal{P}(n)}  \sum_{b \in \mathcal{B}} \sum_{g \in \mathcal{G}} C^{\text{INV}}_{n',g} \Delta G_{n',b,g}.
\end{align}
In other words, investments result in ongoing capital costs throughout the years are covered by the model. This formulation reflects the fact that resources built in the earlier nodes of the model will have completed a larger fraction of their useful lives by the end of the scenario tree. 

At node $n$, the operating cost is 
\begin{align} \label{eq:op_cost_node_n}
    c^{\text{op}}_{n} = & \sum_{b \in \mathcal{B}} \sum_{g \in \mathcal{G}} C^{\text{FIX}}_{g} G_{n,b,g}  +\sum_{b \in \mathcal{B}} \sum_{g \in \mathcal{G}} \sum_{t \in \mathcal{T}} C^{\text{VOM}}_{g} T_t p_{n,b,g,t} + \sum_{b \in \mathcal{B}}  \sum_{t \in \mathcal{T}} \sum_{g \in \mathcal{G}} C^{\text{EN}}_{n,g} T_t p_{n,b,g,t}  \nonumber\\ 
    &+ \sum_{t \in \mathcal{T}} \sum_{i \in \mathcal{I}}\sum_{b \in \mathcal{B}} T_{t} \gamma_i^{\text{PB}} z_{n,b,t,i} + \sum_{t \in \mathcal{T}} \sum_{l \in \mathcal{L}} T_{t} \gamma^{\text{LINE}} sl_{n,l,t},
\end{align}
where the first two terms are the ongoing fixed and variable operation and maintenance costs of generation, the third term is the fuel cost, and the last two terms are penalties for curtailed load and for transmission constraint violations. 

The system planner seeks to maximize the net present value of expected benefits over the assumed scenario tree. The model is formulated as follows:
\begin{subequations}\label{eq:cep}
\begin{alignat}{2}
& \max   & \sum_{n \in \mathcal{N}} \phi_n \zeta_{\delta(n)} \left(
\sum_{t \in \mathcal{T}} \sum_{b \in \mathcal{B}}   T_{t}  \gamma^{\text{LOAD}} D_{n,b,t}  - c^{\text{op}}_{n} - c^{\text{cap}}_{n} \right) \label{eq:cep_obj} \\
& \text{s.t.} 
&L_{n,l} = L_{0,l}+\sum_{n' \in \mathcal{P}(n_{-})} \sum_{q \in \mathcal{Q}} w_{n',l,q}\Delta L_{q} & \quad \forall n \in \mathcal{N} \backslash 0, l \in \mathcal{L} \label{eq:cep_L}\\  
& &G_{n,b,g} = G_{0,b,g}+ \sum_{n' \in \mathcal{P}(n)} \Delta G_{n',b,g} - \sum_{n' \in \mathcal{P}(n)} \Delta \overline{G}_{n
',b,g} & \quad \forall n \in \mathcal{N}, b \in \mathcal{B}, g \in \mathcal{G}  \label{eq:cep_G}\\ 
& (\phi_n \zeta_{\delta(n)}T_t \theta_{n,b,g,t}):  &p_{n,b,g,t} \leq CA_{b,g,t}  G_{n,b,g} & \quad \forall n \in \mathcal{N}, b \in \mathcal{B}, g \in \mathcal{G}, t\in \mathcal{T}  \label{eq:cep_CA}\\
&(\phi_n \zeta_{\delta(n)}\nu_{n}): & \sum_{t \in \mathcal{T}} \sum_{b \in \mathcal{B}} \sum_{g \in \mathcal{G}_R} T_t p_{n,b,g,t} \geq RPS_{n}  \sum_{t \in \mathcal{T}} \sum_{b \in \mathcal{B}}  T_t D_{n,b,t} & \quad \forall n \in \mathcal{N}  \label{eq:cep_RPS} \\
&  (\phi_n \zeta_{\delta(n)}T_t\pi_{n,b,t}):   & NI_{n,b,t} = \sum_{g \in \mathcal{G}} p_{n,b,g,t}  + \sum_{i \in \mathcal{I}}z_{n,b,t,i}-D_{n,b,t} & \quad \forall n \in \mathcal{N}, b \in \mathcal{B}, t\in \mathcal{T}  \label{eq:cep_NI} \\
&    &-(L_{n,l}+sl_{n,l,t})  \leq \sum_{b \in \mathcal{B}'} SF_{l,b} NI_{n,b,t} \leq (L_{n,l}+sl_{n,l,t})  & \quad \forall n \in \mathcal{N}, l \in \mathcal{L}, t\in \mathcal{T}  \label{eq:cep_PF_bound}\\
&    &\sum_{b \in \mathcal{B}}NI_{n,b,t} = 0  & \quad \forall n \in \mathcal{N}, t\in \mathcal{T} \label{eq:cep_power_balance}\\
&    & \Delta G_{n,b,g}, \Delta \overline{G}_{n,b,g}, p_{n,b,g,t} \geq 0  & \quad \forall n \in \mathcal{N}, b\in \mathcal{B}, g \in \mathcal{G}, t \in \mathcal{T} \\
&    &z_{n,b,t,i} \geq 0  & \quad \forall n \in \mathcal{N}, b\in \mathcal{B}, t \in \mathcal{T}, i \in \mathcal{I} \label{eq:cep_load_curtail_seg}\\
&    &\sum_{b \in \mathcal{B}} z_{n,b,t,i} \leq \overline{Z}_{i}  & \quad \forall n \in \mathcal{N}, t \in \mathcal{T}, i \in \mathcal{I} \label{eq:cep_load_curtail_sum}\\
&    & sl_{n,l,t} \geq 0  & \quad \forall n \in \mathcal{N}, l\in \mathcal{L}, t \in \mathcal{T} \label{eq:cep_slack}\\
&    & w_{n,l,q} \in \{0,1\}  & \quad \forall n \in \mathcal{N}, l\in \mathcal{L}, q \in \mathcal{Q}.
\end{alignat}
\end{subequations}
Constraint~\eqref{eq:cep_L} states that the total cumulative transmission capacity is equal to the initial existing transmission capacity plus the sum of the transmission capacity expansion along the path from node~0 to node~$n_{-}$, while constraint~\eqref{eq:cep_G} states that the total cumulative generation capacity is equal to the initial existing generation plus the sum of generation capacity expansion minus generation retirement along the path from node~0 to node~$n$. The delayed in-service date for new transmission relative to new generation is intended to capture the longer development timelines typical for transmission projects. Constraint~\eqref{eq:cep_CA} states that power production is limited by the total installed capacity of a given technology multiplied by its availability in each time block. Constraint~\eqref{eq:cep_RPS} enforces a system-wide renewable portfolio standard (RPS), mandating a percentage of the total amount of power generation coming from renewable energy sources. Constraint~\eqref{eq:cep_NI} calculates the net power injection at bus~$b$, while constraint~\eqref{eq:cep_PF_bound} is a soft constraint limiting power flow on a transmission line. Constraint~\eqref{eq:cep_power_balance} states the sum of the net power injection in the network should be zero. Constraints~\eqref{eq:cep_load_curtail_seg} and~\eqref{eq:cep_load_curtail_sum} state that each load curtailment segment is non-negative and the sum of load curtailment segment cannot exceed the maximum MW violation of that segment. After fixing binary variables $w$, we can query the dual variables of the constraints in the resulting linear program. Dual variables are scaled in order to produced unscaled prices and inframarginal rents. The dual variable $\theta_{n,b,g,t}$ of constraint~\eqref{eq:cep_CA} can be interpreted as the marginal value of capacity of generation technology $g$ at bus $b$ in time block $t$. The dual variable $\pi_{n,b,t}$ of constraint~\eqref{eq:cep_NI} is the locational marginal price (LMP). For completeness we define the linear program using the optimal values $w_{n,l,q}^*$ found when solving model~\eqref{eq:cep} as follows:
\begin{subequations} \label{eq:fixed}
\begin{alignat}{2}
& \max \quad & \eqref{eq:cep_obj} \notag \\
& \text{s.t.} \quad & \eqref{eq:cep_L} \text{\textendash} \eqref{eq:cep_slack} \notag \\
& & \quad w_{n,l,q} = w_{n,l,q}^* & \quad \forall n \in \mathcal{N}, l\in \mathcal{L}, q \in \mathcal{Q}.
\end{alignat}
\end{subequations}

\section{Establishing Beneficiaries} \label{se:beneficiaries}
Supposing that system planners use model~\eqref{eq:cep} to identify transmission expansion decisions, this section addresses the question of how to define beneficiaries, as well as the challenges that arise even when all parties agree on the formulation and scenarios used in the model. 

\subsection{Establishing a Counterfactual} \label{sec:counter}
To measure the benefits brought by a certain transmission project, we first need to define a counterfactual against which benefits will be measured. Establishing a counterfactual to the construction of a particular transmission investment is complicated by the fact that subsequent transmission and generation investment, as well as operations, will change as a result of the investment under study. Some cost allocation schemes currently used in U.S. systems, especially those for investments motivated by reliability violations rather than economic efficiency, fail to establish a valid counterfactual because the models omit the possibility of operational changes or compensatory investments. As discussed in~\cite{Mays2023}, the absence of a valid counterfactual is particularly clear in the case of interconnecting new generators.

After solving model~\eqref{eq:cep} and determining expansion decisions for the present year, there are at least three ways that a counterfactual might be established. In each case, re-solving model~\eqref{eq:cep} with additional constraints leads to an alternate solution with a higher objective function value. We define three options as follows:
\begin{enumerate}
    \item Exclude the specific transmission investment and fix all other transmission and generation investments; benefits reflect the difference in operating cost between the solutions.
    \item Exclude the specific transmission investment, fix all other transmission investments, and allow generation investments to optimally readjust to the counterfactual network; benefits reflect the difference in investment and operating cost between the solutions.
    \item Exclude the specific transmission investment (at all levels $q\in\mathcal{Q}$ and for either all years $y\in\mathcal{Y}$ or just the present), but allow freedom in both generation and other transmission investments; benefits reflect the difference in investment and operating cost between the solutions. 
\end{enumerate}
The primary issue with the first option is that it is unrealistic and unnecessarily restrictive. Excluding the transmission investment without allowing any compensatory investments could lead to a situation with unsolvable reliability violations, leading either to an infeasible model or large costs driven by penalty parameters. The primary issue with the third option is that in order to determine participant-level benefits for the projects of interest, cost allocation determinations also need to be made for the counterfactual transmission projects. Since these allocations would in turn be determined against a similarly defined counterfactual, allowing these alternatives introduces a recursive aspect to the problem. Since allocations based on the first are guaranteed to be inaccurate and those based on the third would be impractical, we suggest that analysis should pursue the second option. Since investment in generation (as well as storage and distributed resources) is often exogenous or excluded from current models, we note the contrast between our recommendation and the claim in~\cite{Hogan2018} that the information needed for cost allocation is already available in current planning models. 

Putting this suggestion into practice could be challenging, especially in the case of upgrades prompted by reliability violations not observable in the linear approximations to the power flow equations typically used in capacity expansion models. At the expense of additional complexity, more complicated constraints could in principle be brought into model~\eqref{eq:cep}, making the construction of a valid counterfactual more straightforward. In practice, it is more common in such cases to skip the step of establishing a counterfactual altogether, instead socializing the cost of related upgrades or relying on power flow analyses with unclear connection to economic benefits. Even if an optimization model is not used, however, a better approach to assess benefits would be specifying a plausible alternative to resolve the identified reliability violations and measuring cost against this alternative. Noting the challenge, for the remainder of this paper we assume that benefits associated with the level of reliability are captured through penalties on power balance violations in model~\eqref{eq:cep}.


We can formalize the construction of a counterfactual as follows. Suppose we are interested in allocating the cost of one or more transmission investments represented by a subset $\mathcal{W}^{INV}\subset \mathcal{W} = \mathcal{N}\times\mathcal{L}\times\mathcal{Q}$ of the binary variables $w_{n,l,q}$, where we assume that the investments of interest occur at node $n=0$. Then counterfactual generation investments, along with counterfactual prices and production quantities, can be found by solving
\begin{subequations} \label{eq:counter}
\begin{alignat}{2}
& \max \quad & \eqref{eq:cep_obj} \notag \\
& \text{s.t.} \quad & \eqref{eq:cep_L} \text{\textendash} \eqref{eq:cep_slack} \notag \\
& & \quad w_{n,l,q} = 0 & \quad \forall (n,l,q) \in \mathcal{W}^{INV} \\
& & \quad w_{n,l,q} = w_{n,l,q}^* & \quad \forall (n,l,q) \in \mathcal{W} \setminus \mathcal{W}^{INV}.
\end{alignat}
\end{subequations}

\subsection{Generation} \label{sec:gen-benefits}
We first consider the potential for generators to benefit from transmission expansion. An important distinction is between new and existing generators. At present, most U.S. systems follow an ``invest and connect'' approach in which the cost of network upgrades identified in interconnection studies is allocated to new generators. ERCOT, by contrast, uses a ``connect and manage'' approach that eschews network upgrades in the interconnection process but does not make any guarantees on the deliverability of energy from the interconnecting project. At present, no U.S. system allocates cost for subsequent network upgrades to generators after they have completed the interconnection process. The primary point of this subsection is to show that, in general, the direct benefits modeling approach pursued in this paper supports the ``connect and manage'' approach of allowing new generators to join the system without cost, but also supports the allocation of cost to existing generators on an ongoing basis.

Evaluated at node~0, the discounted operating profit expected by a unit of generation of type $g$ at bus $b$ can be calculated as
\begin{align} 
\label{eq:unit_op_profit}
    \mathbb{E}(\mathcal{U}_{b,g}^{gen}) &= \sum_{n \in \mathcal{N}} \zeta_{\delta (n)} \phi_n \left(\sum_{t \in \mathcal{T}} T_t  \left(\pi_{n,b,t} -  C^{\text{EN}}_{n,g}- C^{\text{VOM}}_{g} + \nu_n \mathbbm{1}_{\{g\in \mathcal{G}_R\}} \right) \frac{p_{n,b,g,t}}{G_{n,b,g}} - C^{\text{FIX}}_{g}  \right),
    \end{align}
where $\mathbbm{1}_{\{g \in \mathcal{G}_R\}} = 1$ if the generator can sell renewable energy credits and $0$ otherwise. 

With $\mathbb{E}(\mathcal{U^*}_{b,g}^{gen})$ indicating expected benefits assuming the socially optimal transmission configuration and $\mathbb{E}(\mathcal{U'}_{b,g}^{gen})$ indicating expected benefits with the counterfactual transmission configuration, the per unit expected benefit for generation of type $g$ located at bus $b$ from transmission expansion can then be calculated as the difference in expected operating profits:
\begin{align} \label{eq:gen_benefits}
     \mathbb{E}(\Delta \mathcal{U}_{b,g}^{gen}) &= \mathbb{E}(\mathcal{U^*}_{b,g}^{gen})  - \mathbb{E}(\mathcal{U'}_{b,g}^{gen}).
\end{align}

\subsubsection{Existing generators}
We first consider the case of existing generators, which can more clearly benefit or be harmed by transmission expansion. The presence of new generation in either the socially optimal or the counterfactual case can indicate how existing generators of the same type and located at the same bus are affected by the expansion. We state three cases formally as Theorem~\eqref{theo:gen-nogain} and Corollaries~\eqref{prop:gen-gain} and~\eqref{prop:gen-loss}.
\begin{theo} \label{theo:gen-nogain}
 Suppose new generation of type $g$ is constructed at bus $b$ in both the expansion scenario, i.e., $\Delta G^*_{0,b,g} > 0$, and the counterfactual scenario, i.e., $\Delta G'_{0,b,g} > 0$. Then existing generation of that type at that bus neither benefits nor suffers losses from the expansion, i.e., $\mathbb{E}(\Delta \mathcal{U}_{b,g}^{gen}) = 0$.
\end{theo}

\begin{proof}
For model~\eqref{eq:fixed}, the KKT conditions on $p_{n,b,g,t}$ are 
\begin{align}
0 \leq p_{n,b,g,t} \perp C^{\text{EN}}_{n,g} + C^{\text{VOM}}_{g}  + \theta_{n,b,g,t}
    - \pi_{n,b,t} - \nu_n \mathbbm{1}_{\{g \in \mathcal{G}_R\}} \geq 0 \qquad \forall n \in \mathcal{N}, b \in \mathcal{B}, g \in \mathcal{G}, t \in \mathcal{T}.  \label{eq:KKT_gen_prod} 
\end{align}

By the complementarity condition, if $p_{n,b,g,t}>0$, we have $\theta_{n,b,g,t} = \pi_{n,b,t} - C^{\text{EN}}_{n,g} - C^{\text{VOM}}_{g} + \nu_n \mathbbm{1}_{\{g \in \mathcal{G}_R\}} $. 

Then the discounted operating profit~\eqref{eq:unit_op_profit} can be written as
\begin{align}
    \mathbb{E}(\mathcal{U}_{b,g}^{gen})
    &= \sum_{n \in \mathcal{N}} \zeta_{\delta (n)} \phi_n \left(\sum_{t \in \mathcal{T}} T_t  \theta_{n,b,t} \frac{p_{n,b,g,t}}{G_{n,b,g}} - C^{\text{FIX}}_{g}  \right).
\end{align}

By complementary slackness, when $\theta_{n,b,t} >0$, $\frac{p_{n,b,g,t}}{G_{n,b,g}} = CA_{b,g,t}$ holds. When $\theta_{n,b,t} = 0$, replacing $\frac{p_{n,b,g,t}}{G_{n,b,g}}$ with $CA_{b,g,t}$ would not affect the result. After replacement, the discounted operating profit~\eqref{eq:unit_op_profit} becomes
\begin{align}
    \mathbb{E}(\mathcal{U}_{b,g}^{gen})= \sum_{n \in \mathcal{N}} \zeta_{\delta (n)} \phi_n \left(\sum_{t \in \mathcal{T}} T_t  \theta_{n,b,t} CA_{b,g,t} - C^{\text{FIX}}_{g}  \right).
\end{align}

For both models~\eqref{eq:fixed} and~\eqref{eq:counter}, the objective function and variable $G_{n,b,g}$ are defined to include summation over the path $\mathcal{P}(n)$. Given that node $0$ is on the path of every node to the root node in the scenario tree, it follows that the KKT condition on $\Delta G_{0,b,g}$ would aggregate over all nodes within the tree, given by
\begin{align} \label{eq:KKT_inv}
    0 \leq \Delta G_{0,b,g}  \perp \sum_{n \in \mathcal{N}}  \zeta_{\delta (n)} \phi_n \left( C^{\text{INV}}_{0,g} + C^{\text{FIX}}_g - \sum_{t \in \mathcal{T}}T_t CA_{b,g,t} \theta_{n,b,g,t}\right)  \geq 0 \qquad \forall g \in \mathcal{G}, b \in \mathcal{B}. 
\end{align}

By complementary slackness, $\Delta G_{0,b,g} > 0$ implies $$ \sum_{n \in \mathcal{N}}  \zeta_{\delta (n)} \phi_n \left( C^{\text{INV}}_{0,g} + C^{\text{FIX}}_g - \sum_{t \in \mathcal{T}}T_t CA_{b,g,t} \theta_{n,b,g,t}\right) = 0.$$ 
When new generation of type $g$ is constructed at bus $b$ in the both expansion scenario and counterfactual scenario, i.e., $\Delta G^*_{0,b,g} > 0$ and $\Delta G'_{0,b,g} > 0$, we have
$\mathbb{E}(\mathcal{U^*}_{b,g}^{gen}) = \mathbb{E}(\mathcal{U'}_{b,g}^{gen}) = \sum_{n \in \mathcal{N}}  \zeta_{\delta (n)} \phi_n C_{0,g}^{\text{INV}}$. By the definition of benefits in~\eqref{eq:gen_benefits}, this leads to $\mathbb{E}(\Delta \mathcal{U}_{b,g}^{gen}) = 0$.

\end{proof}

\begin{coro} \label{prop:gen-gain}
 Suppose new generation of type $g$ is constructed at bus $b$ in the expansion scenario, i.e., $\Delta G^*_{0,b,g} > 0$, but not in the counterfactual scenario, i.e., $\Delta G'_{0,b,g} = 0$. Then existing generation of that type at that bus benefits from the expansion, i.e., $\mathbb{E}(\Delta \mathcal{U}_{b,g}^{gen}) > 0$.
\end{coro}
\begin{proof}
    As shown in Theorem~\ref{theo:gen-nogain}, $\Delta G^*_{0,b,g}>0$ implies $\mathbb{E}(\mathcal{U}_{b,g}^{gen}) = \sum_{n \in \mathcal{N}}  \zeta_{\delta (n)}\phi_n C_{0,g}^{\text{INV}}$. $\Delta G'_{0,b,g}=0$ implies $\mathbb{E}(\mathcal{U'}_{b,g}^{gen}) < \sum_{n \in \mathcal{N}}  \zeta_{\delta (n)}\phi_n C_{0,g}^{\text{INV}}$. Therefore, by~\eqref{eq:gen_benefits}, the difference in expected operating profits is positive, i.e., $\mathbb{E}(\Delta \mathcal{U}_{b,g}^{gen}) > 0$.
\end{proof}

\begin{coro} \label{prop:gen-loss}
 Suppose new generation of type $g$ is constructed at bus $b$ in the counterfactual scenario, i.e., $\Delta G'_{0,b,g} > 0$, but not in the expansion scenario, i.e., $\Delta G^*_{0,b,g} = 0$. Then existing generation of that type at that bus suffers losses from the expansion, i.e., $\mathbb{E}(\Delta \mathcal{U}_{b,g}^{gen}) < 0$.
\end{coro}

\begin{proof}
    As shown in Theorem~\ref{theo:gen-nogain}, $\Delta G^*_{0,b,g}=0$ implies $\mathbb{E}(\mathcal{U}_{b,g}^{gen}) < \sum_{n \in \mathcal{N}}  \zeta_{\delta (n)}\phi_n C_{0,g}^{\text{INV}}$. $\Delta G'_{0,b,g}>0$ implies $\mathbb{E}(\mathcal{U'}_{b,g}^{gen}) = \sum_{n \in \mathcal{N}}  \zeta_{\delta (n)}\phi_n C_{0,g}^{\text{INV}}$. Therefore, by~\eqref{eq:gen_benefits}, the difference in expected operating profits is negative, i.e., $\mathbb{E}(\Delta \mathcal{U}_{b,g}^{gen}) < 0$.
\end{proof}

At a high level, it can be expected that generators in exporting regions will see benefits from transmission expansion while generators in importing regions will suffer losses. Clear examples of this effect are shown in~\cite{Hogan2018} and~\cite{ONeill2020}, which analyze two-zone systems without subsequent generation investment. The more complex numerical study in this paper largely matches this intuition.


\subsubsection{New generators}
We now turn attention to newly built generation. If these resources would have been built in the model even without the transmission expansion occurring in node~0, then benefits can be defined similarly to existing generators. In this case, Theorem~\ref{theo:gen-nogain} applies and we conclude that the new generation does not benefit from the transmission. If the generation would not otherwise be built, the zero-profit condition on investment in the socially optimal expansion nevertheless holds. Given perfect competition, condition~\eqref{eq:KKT_inv} implies that investment in generation technology $g$ will continue until operating profits fall to the level of annualized investment costs. 

Under an optimization modeling approach, the implication of the zero-profit condition is that new generation cannot be identified as a beneficiary. We note that the assumptions of perfect competition and linear generation investment costs that underpin the zero-profit condition are standard in tools used for expansion planning. Exceptions to this rule may apply, e.g., if there is a constraint on building generation, such that new capacity cannot be built to take full advantage of the new line. In this case, new generators would earn a rent associated with this constraint. Such exceptions are likely to be less important for large lines that would facilitate production across a wider region. Alternatively, it may be argued that the computation of excess profit in Eq.~\eqref{eq:gen_benefits} reflects too narrow a conception of benefits, and the existence of a new generator could itself be considered a benefit regardless of its profitability. In this case, additional assumptions outside the planning tool would be needed to define benefit estimates.

Given the recommendation not to allocate cost for network upgrades to new generators at the time of interconnection but then subsequently allocate cost to them throughout their life, the direct benefits modeling approach supports a significant change to current practice. The overall impact that such a change would have on the cash flows seen by generators over the course of their life is not clear. Suppose a new generator signed an interconnection agreement under a ``connect and manage'' approach, and then welfare-enhancing network upgrades were identified by a planning model. Since the model recommending these projects would assume the presence of the new generator, it would be more likely to recommend upgrades allowing the system to make use of the new generator's energy output. Projects identified to make use of the new generator could very well be the same as those that would be identified under the narrower models used in ``invest and connect'' interconnection procedures. Whereas current practice typically assigns the cost entirely to the new generator without accounting for externalized benefits, however, the proposed approach would allocate cost to other beneficiaries as well. As such, the overall effect would be to bring cost allocation in line with the beneficiaries pay principle.

\subsection{Load}

Benefits to different load zones can be defined much in the same way as benefits to generators. The major difference is that the planning model takes load as exogenous rather than as resulting from an expansion decision that may depend on transmission investment. Consumer surplus at bus $b$ in node $n$ is calculated as the difference between the value of energy consumed and payments made for energy and renewable energy credits. Since we are primarily interested in allocating cost between different zones, each of which comprises a diverse range of customers, it is reasonable to assume a single constant for the value of energy. However, we note that a more detailed representation of price-responsive load for individual customers would enable a more granular calculation of benefits. To simplify notation, we represent load curtailment at bus $b$ at time $t$ by $z_{n,b,t}$ where $z_{n,b,t} = \sum_{i \in \mathcal{I}}z_{n,b,t,i}$. 
Evaluated at node~0, the expected value of consumer surplus can be written as
\begin{equation} \label{eq:consumer}
        \mathbb{E}(\mathcal{U}_b^{load}) = \sum_{n \in \mathcal{N}} \zeta_{\delta (n)} \phi_n \left(\sum_{t \in \mathcal{T}} T_t  \left( \gamma^{\text{LOAD}} - \pi_{n,b,t} - \nu_n RPS_n \right) \left(D_{n,b,t} - z_{n,b,t}\right) \right). 
\end{equation}

As with generation, we compute the benefits from transmission expansion to loads at bus $b$ as the difference in surplus between the socially optimal case and the counterfactual case:
\begin{equation} \label{eq:load_benefit}
    \mathbb{E}(\Delta \mathcal{U}_b^{load}) = \mathbb{E}(\mathcal{U^*}_b^{load})  - \mathbb{E}(\mathcal{U'}_b^{load}).
\end{equation}

As with existing generators, the surplus difference $\mathbb{E}(\Delta \mathcal{U}_b^{load})$ can be positive, zero, or negative for loads.

\subsection{Congestion Rents}

In addition to generator and consumer benefits, a third component of market surplus is transmission congestion rents, computed as the difference between the payments made by load and the revenue received by generators. Under idealized assumptions, the availability of congestion rents could make the problem of cost allocation easier to solve, since congestion revenues would be sufficient to support a socially efficient level of transmission expansion~\citep{Hogan2018,joskow2005merchant}. In this case a ``top-down'' cost allocation would not be required as such, since risk-neutral investors would be willing to build transmission in exchange for the resulting valuable transmission rights. In practice, economies of scale and unpriced reliability constraints mean that congestion rents are well below what would be needed to support an efficient level of investment. For example,~\cite{Sherman2023} estimates that U.S.-wide congestion rents averaged approximately~\$8.2B for 2016\textendash 2021, while~\cite{EIA-AEO} estimates the average cost of transmission in 2022 at \$15/MWh, implying a total annualized cost of roughly~\$63B for the current system. While imprecise, these estimates suggest that congestion rents are an order of magnitude lower than what would be required for investments in transmission to be sustained on a merchant basis.

In principle, rights to congestion rents can be allocated as part of the cost allocation process, either proportional to market participant contributions or by auction. In general, empirical evidence in U.S. markets shows that current markets for financial transmission rights result in large transfers from consumers to financial traders~\citep{Leslie2021,Opgrand2022}, suggesting opportunities for improvements in allocation~\citep{Risanger2024}. In the numerical results for this paper, we compute generator and consumer benefits without adjusting for any assigned transmission rights, noting that future studies assessing the effects of financial transmission rights would likely require downscaling the results of our zonal network model to a more detailed nodal representation.



\subsection{Multi-value Planning}

Inconsistent and non-intuitive cost allocation outcomes in the U.S. context can stem from projects being designated as having a single primary purpose and being evaluated according to the benefits it provides only along that dimension. U.S. systems distinguish between projects undertaken for economics, reliability, public policy, and generator interconnection, while any transmission enhancement necessarily affects outcomes across all four areas~\citep{DeLosa2024}. 
As previously noted, we leave a more complete discussion of public policy interactions for future work. We note here, however, that an advantage of the direct benefits modeling approach is that all categories of benefits can be incorporated in a consistent manner as long as a valid counterfactual can be established. From a modeling perspective, the only requirement for establishing a valid counterfactual is that model~\eqref{eq:counter} must have a feasible solution after transmission expansion decisions have been fixed. Because the model penalizes power shortfalls rather than implementing a hard constraint, and because entry of new generation is not restricted, reliability constraints cannot cause infeasibility. Our implementation of an RPS in Eq.~\eqref{eq:cep_RPS} could in principle lead to infeasibility. In practice, however, most states have implemented Alternative Compliance Payments to limit the potential cost of RPS policies, meaning that a soft constraint would more accurately reflect the public policy. Once a counterfactual has been established, costs associated with reliability and public policy naturally flow into prices for energy and clean attributes, allowing straightforward inclusion in benefit calculations.

\subsection{Cost Allocation} \label{se:cost_allocation}
The analysis thus far leads to the conclusion that existing generation and load can be beneficiaries of transmission expansion over the long term, implying that both existing generation and load should share cost under the ``beneficiaries pay'' principle. In the numerical study we examine two different policies for allocating the cost of transmission investments made at node~0: allocating cost only to load, as in current practice, and allocating across both existing generation and load.    

When allocating cost only to load, the allocation ratio to load at bus $b$ is determined using the following equation:
\begin{align} \label{eq:load_only_allocation}
    r_b^{load}=\frac{[\mathbb{E}(\Delta \mathcal{U}_b^{load})]_{+}}{\sum_{b'\in \mathcal{B}} [\mathbb{E}(\Delta \mathcal{U}_{b'}^{load})]_{+}},
\end{align}
where $[*]_{+}$ denotes $\max\{0,*\}$. 

When allocating cost to both load and existing generation, with $G_{b,g}^0$ representing the quantity of existing capacity of generation $g$ at bus $b$, allocation ratios are determined using the following equations:
\begin{subequations} \label{eq:cost_allocation}
\begin{align} 
    &r_b^{load}=\frac{[\mathbb{E}(\Delta \mathcal{U}_b^{load})]_{+}}{\sum_{b'\in \mathcal{B}} \left([\mathbb{E}(\Delta \mathcal{U}_{b'}^{load})]_{+} + \sum_{g \in \mathcal{G}} [G_{b,g}^0 \mathbb{E}(\Delta \mathcal{U}_{b',g}^{gen})]_{+} \right)}; \label{eq:load_allocation}\\
    &r_{b,g}^{gen}=\frac{[G_{b,g}^0 \mathbb{E}(\Delta \mathcal{U}_{b,g}^{gen})]_{+}}{\sum_{b'\in \mathcal{B}} \left([\mathbb{E}(\Delta \mathcal{U}_{b'}^{load})]_{+} + \sum_{g \in \mathcal{G}} [G_{b,g}^0\mathbb{E}(\Delta \mathcal{U}_{b',g}^{gen})]_{+} \right)}. \label{eq:gen_allocation}  
\end{align}
\end{subequations}


The presence of the max operator ensures that market participants who do not benefit from a transmission investment are not allocated cost. However, it also implies that market participants that are harmed by an expansion project are not compensated as part of cost allocation. It would be straightforward mathematically to define negative allocations, i.e., compensatory payments to these participants. Since the planning model by assumption identifies a surplus-maximizing expansion plan, there would be sufficient surplus in the market to make these compensatory payments. Several recent cases in the U.S. show the potential for states or incumbents hurt by transmission expansion to intervene and prevent it from occurring (see, e.g.,~\cite{Hausman2024}), suggesting that compensatory payments or long-term financial rights that protected incumbents against the effect of transmission expansion could lead to fewer disputes in planning.


\section{Numerical Study} \label{se:example}
This section presents results of a numerical example on a simplified model of the ERCOT system. Building on the discussion of Section~\ref{se:beneficiaries}, we document the different benefits and losses seen by generation and loads in different parts of the system. One major conclusion of the numerical study is that allocating cost on a portfolio basis is likely to be more consistent with the beneficiaries pay principle than allocating on a project-by-project basis. Further, we contrast ex post benefits derived from out-of-sample tests against in-sample estimation, computing the range of possible distributional outcomes from transmission expansion to provide insight into the challenge posed by ambiguity in future scenarios and probabilities.

\subsection{Data and Study Assumptions} \label{se:case}
The study employs an 8-Bus ERCOT DC Test Case introduced by~\cite{battula2020ercot}, with the network shown in Figure~\ref{fig:texas}.
\begin{figure}[!h]
    \centering
    \includegraphics[width = 10cm]{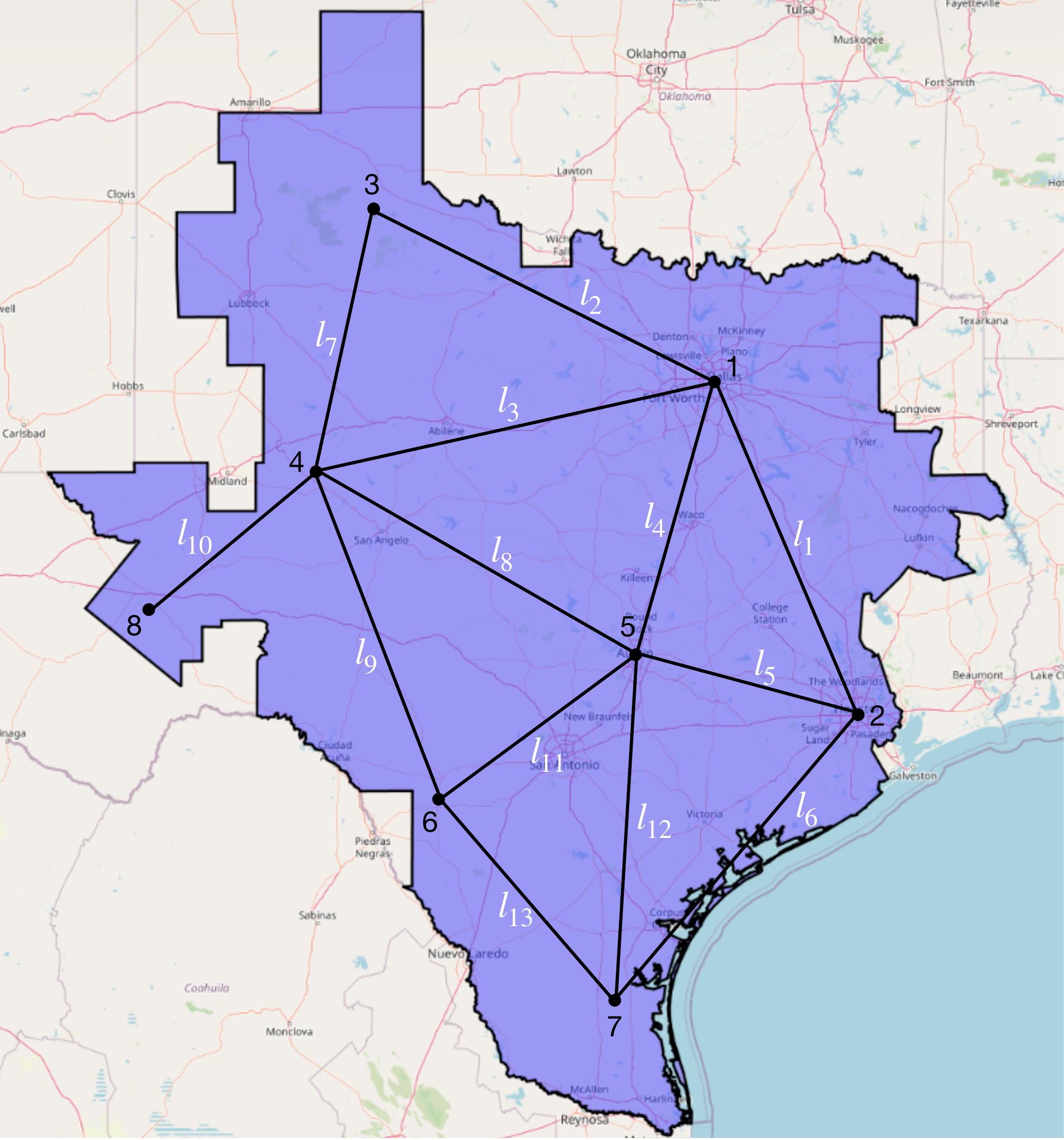}
    \caption{8-Bus ERCOT network.}
    \label{fig:texas}
\end{figure}
The generation technologies considered are natural gas combined cycle (CC), natural gas combustion turbine (CT), coal, nuclear, utility-scale solar, and land-based wind. Costs for these technologies are sourced from the NREL Annual Technology Baseline database~\citep{nrel2022annual}. The existing generation capacity mix is obtained from the ERCOT Capacity, Demand and Reserves (CDR) Report~\citep{ERCOT_CDRreport}. Existing generation capacity, reported in Table~\ref{tab:exist_gen}, is assigned to different buses in the test system in a manner consistent with the ERCOT resource siting methodology report~\citep{ERCOT_LSTA,ERCOT_LSTA_siting}, but should not be expected to match locations precisely.

\begin{table}[] \caption{Existing capacity of each generation type by bus.}
    \label{tab:exist_gen}
    \centering 
    \begin{tabular}{c|cccccccc}
    \hline
Generation type &	$b_1$ &		$b_2$ &		$b_3$ &		$b_4$ &		$b_5$	&	$b_6$	&	$b_7$ &		$b_8$ \\ \hline
Gas CC	&	0	&	6,062	&	12,644	&	1,839	&	0	&	0	&	8,282	&	0 \\
Gas GT	&	0&		10,404&		6,011&		2,570&		0&		0	&	9,343&		0\\
Coal	&	0	&	2,514	&	7,023&		0&		4,031&		0&		0&		0\\
Nuclear	&	2,400& 2,573	&	0	&	0	&	0	&	1,030	&	0	&	0\\
Solar&		0	&	468&		1,073&		0	&	0	&	0&		850&		6,611\\
Wind&	0&	4,865&		1,330&		17,291 &0	&	0	&	2,969	&	0\\ \hline
    \end{tabular}
\end{table}

Hourly load data for the year 2020 from~\cite{ERCOTload} is used to represent the load profiles in the system. For each node $n$, load $D_{n,b,t}$ is obtained by multiplying this profile by a demand growth factor $\beta_{n}$. Hourly solar and wind availability profiles for the year 2020 are extracted from~\cite{ninja} using methods described in~\cite{pfenninger2016long} and~\cite{staffell2016using}. To ensure computational tractability and account for operation costs, a K-means method is employed to cluster the year of data based on the net load, from which $20$ representative days ($480$ hours) with varying weight, i.e., $T_t$, are selected to represent a simulation year.

In the long-term planning model, the uncertainties included are the presence of an RPS, load growth, technology investment costs for wind and solar, and fuel cost. For each uncertainty except the RPS, low, medium, and high values are estimated based on \cite{nrel2022annual, ERCOT_report}. We note that given the high-quality solar and wind resources in Texas, the RPS constraint does not have a significant impact on the numerical results. A future scenario is defined as a subset of the uncertainty space that represents a specific combination of the five uncertainties. Considering a low, medium, and high value for each uncertainty, there will be a total of $3^5=243$ possible future scenarios. To ensure computational tractability for the MIP model, the number of scenarios must be reduced. In this study, seven scenarios with varying probabilities were selected based on the methodology described in~\cite{newlun2021adaptive}. 
Since we wish to avoid making assumptions on underlying scenario probabilities, we do not claim that the transmission plan identified by the model is ``optimal'' as such. Out-of-sample tests show positive net benefits in all scenarios, however, suggesting that the chosen clustering and scenario selection procedures lead to a high-quality solution.

A 20-year planning horizon is simulated with investment decisions made every 5 years, resulting in a tree with depth four and seven scenarios. A discount rate of $7.78\%$ is applied to compute the net present value of the total investment cost and operational cost in the objective function. It is assumed that the operational costs for each successive 5-year interval remain constant. In light of this assumption, the discount factor, denoted as $\zeta_y$ for time index $y \in \{1,2,3,4\}$, is determined through the following formula: $$\zeta_y = \frac{1}{(1+7.78\%)^{5(y-1)}}\left ( 1+\frac{1}{(1+7.78\%)^1}+\frac{1}{(1+7.78\%)^2}+\frac{1}{(1+7.78\%)^3}+\frac{1}{(1+7.78\%)^4}\right).$$ We use the power balance penalty curve shown in Table~\ref{tab:PB} and transmission line violation penalty $\gamma^{\text{LINE}} = 9251$ \$/MW for all lines, congruent with the practices in ERCOT~\cite{ercot_pb}.

\begin{table}[H]
    \centering
    \caption{ERCOT power balance penalty curve.}
    \begin{tabular}{c|ccccccccc}
\hline
    MW violation & $\leq 5$ & $5 \sim 10$  & $10 \sim 20$ & $20 \sim 30$ & $30\sim40$ & $40\sim50$ & $50\sim100$ & $\geq 100$ \\
    $\gamma_i^{\text{PB}}$ (\$/MWh) & $250$ & $300$ & $400$ & $500$ & $1000$ & $2250$ & $4500$ & $5001$\\
    \hline
\end{tabular}
    \label{tab:PB}
\end{table}

This case study assumes that there are seven types of transmission line expansion increments, with the same costs across all scenarios in each stage and for each corridor, as defined in Table~\ref{tab:line_inv_type}. The per unit investment cost in Table~\ref{tab:line_inv_type} exhibits a significant decrease with increasing expansion capacity, reflecting economies of scale. 
 
\begin{table}[H]
    \centering
    \caption{Transmission line capacity increment type and investment cost} \label{tab:line_inv_type}
    \begin{tabular}{c|c|c|c}
    \hline
        Type & Expansion (MW) & Amortized investment cost (\$M/yr) & Per unit cost (\$M/MW) \\
        \hline
        1 & 1400 & 68.93 &0.61\\
        2 & 1800 & 72.64 &0.50\\
        3 & 2300 & 78.34 &0.42\\
        4 & 3000 & 89.59 &0.37\\
        5 & 3600 & 98.79 &0.34\\
        6 & 4200 & 101.70 &0.30\\
        7 & 8000 & 154.96 &0.24\\
    \hline    
    \end{tabular}
\end{table}


The models are implemented in Julia~\citep{Bezanson2017Julia} using JuMP.jl~\citep{DunningHuchetteLubin2017JuMP} and solved with Gurobi version 10.0.1~\citep{gurobi} using a MIP gap (where applicable) of $0.5\%$. The computations are performed on a Mac computer with an Apple M1 Max chip and 10 cores.

\subsection{Expansion Plan}

The transmission line expansions resulting from model~\eqref{eq:cep} for the developed test system are summarized in Table~\ref{tab:add_line_inv}. 
\begin{table}[!ht]
    \centering
    \caption{Invested line capacity (MW) by scenario and by stage. Some scenarios with no line investment are omitted.}
    \label{tab:add_line_inv}
    \begin{tabular}{cc|cccccccccccccc}
    \hline
        Year & Scenario & $l_1$ & $l_2$ & $l_3$ & $l_4$ & $l_5$ & $l_6$ & $l_7$ & $l_8$ & $l_9$ & $l_{10}$ & $l_{11}$ & $l_{12}$ & $l_{13}$ \\ \hline
       2023& -- & 0 & 8000 & 2300 & 0 & 0 & 1800 & 3600 & 0 & 0 & 2300 & 0 & 2300 & 0 \\ 
       2028& 7 & 0 & 0 & 0 & 0 & 0 & 0 & 0 & 3000 & 0 & 0 & 0 & 0 & 0\\
       2033& 1,2,6 & 0 & 0 & 0 & 0 & 0 & 0 & 0 & 3000 & 0 & 0 & 0 & 0 & 0\\
       2033& 3,4 & 0 & 0 & 0 & 0 & 0 & 0 & 0 & 2300 & 0 & 0 & 0 & 0 & 0\\
       2033& 5 & 0 & 0 & 0 & 0 & 0 & 0 & 0 & 3600 & 0 & 0 & 0 & 0 & 0\\
       2038& 3 & 0 & 0 & 1800 & 1800 & 0 & 0 & 0 & 4200 & 4200 & 0 & 0 & 0 & 0\\
       \hline
    \end{tabular}
\end{table}
In the first stage (year $2023$), corresponding to node~$0$ of the tree in Figure~\ref{fig:scenario}, six transmission expansion projects are selected. The largest of these is on path $l_2$, connecting the generation-rich zone $b_3$ with the population center $b_1$.

Table~\ref{tab:add_gen_inv} shows a weighted average of the total generation capacity additions made over the 20-year horizon in the seven modeled scenarios, as well as the same quantities in a counterfactual without any transmission expansion. In either case, the model builds new gas combustion turbines, solar, and wind, with no additions of coal, nuclear, or combined cycle gas in any scenario. 
\begin{table}[!ht]
    \centering
    \caption{Average 20-year generation capacity investment across seven modeled scenarios (MW)}
    \label{tab:add_gen_inv}
    \begin{tabular}{ccc}
    \hline
        Generation Type & Expansion & No Expansion \\ \hline
       Gas CC & 0 & 0 \\ 
       Gas CT & 21,425 & 35,428 \\
       Coal& 0 & 0\\
       Nuclear& 0 & 0\\
       Solar& 9,859 & 14,938 \\
       Wind& 38,382 & 38,433 \\
       \hline
    \end{tabular}
\end{table}
While we expected new transmission to support the deployment of additional wind, the primary effect of expansion in our model was instead to reduce the requirement for gas turbines: the model builds roughly the same amount of wind in the counterfactual case, but substantially more gas turbines. At node~0 of the model, new gas turbines are built at $b_1$ in the expansion scenario, while new gas turbines are built at $b_1$ and $b_8$ in the counterfactual. No wind or solar is built in node~0, potentially due to our assumption that transmission expansions will enter service only in the second year index.

\subsection{Portfolio vs. Project-by-Project Allocation} \label{se:port-proj}

The first policy question we address is whether to assess the benefit of the six projects selected at node~0 on a portfolio or project-by-project basis. In our notation, the question is whether to compute a single instance of the counterfactual model~\eqref{eq:counter} with $\mathcal{W}^{INV}$ including all six projects as a portfolio, or six separate instances of model~\eqref{eq:counter} with a single element each in $\mathcal{W}^{INV}$. For purposes of this subsection, we allocate costs according to Eq.~\eqref{eq:load_only_allocation}, i.e., only to loads and only to those with positive benefits, without compensating those harmed by the expansion. We conclude that assessment at the portfolio level results in a cost allocation more consistent with the beneficiaries pay principle.

Table~\ref{tab:project_allocation} shows the benefits calculated for each project evaluated separately as well as the portfolio, along with an allocation percentage across the~8 buses. The first observation is that, for each individual project except the expansion on $l_6$, the total expected benefit across all buses is negative. As a consequence, the projects would not pass a benefit\textendash cost test when assessed as individual projects, despite being part of a beneficial portfolio of projects. While not guaranteed, each project has at least one bus with positive benefits, allowing a cost allocation to be defined under our formula.
\begin{table}[!ht]
    \centering 
    \caption{Expected nodal benefits and cost allocation ratios when allocating solely to loads. For ``Projects Sum,'' the allocation ratio for each bus is calculated based on the sum of allocated cost across all projects calculated individually. For ``Portfolio,'' the allocation ratio for each bus is calculated from portfolio benefits.}
    \label{tab:project_allocation}
    \begin{tabular}{ll|ccccccccc}
    \hline
        Project &  & $b_1$ & $b_2$ & $b_3$ & $b_4$ & $b_5$ & $b_6$ & $b_7$ & $b_8$  & Sum\\ \hline
        & load ratio (\%) & 33.95 & 28.67 & 1.83 & 1.85 & 15.89 & 0.93 & 8.21 & 8.67 &100.0\\ \hline
        \multirow{2}{*}{$l_2$ 8000 MW} & $\Delta \mathcal{U}_b$ (\$M)  & 4,904 & -1,488 & -1,509 & -171 & -831 & -15 & -2,015  & 411  & -715 \\ 
        & $r_b$ (\%)  & 92.27 & 0.0 & 0.0 & 0.0 & 0.0 & 0.0 & 0.0 & 7.73 & 100.0 \\ \hline
        \multirow{2}{*}{$l_3$ 2300 MW} & $\Delta \mathcal{U}_b$ (\$M)  & 3,668 & -2,279 & -61 & -625 & -2,233 & -104 & -908 & -2,170 & -4,710\\ 
        & $r_b$ (\%)  & 100.0 & 0.0 & 0.0 & 0.0 & 0.0 & 0.0 & 0.0 & 0.0 & 100.0  \\ \hline
        \multirow{2}{*}{$l_6$ 1800 MW} & $\Delta \mathcal{U}_b$ (\$M)  & -67 & 2,600 & -2 & 1 & 53 & -9 & -1,887 & 132 & 822\\ 
        & $r_b$ (\%)  & 0.0 & 93.32 & 0.0 & 0.04 & 0.0 & 0.0 & 0.0 & 1.90  & 100.0 \\\hline
        \multirow{2}{*}{$l_7$ 3600 MW} & $\Delta \mathcal{U}_b$ (\$M)  & -456 & -1,714 & -912 & 304 & 85 & -159 & -1,525 & 795 &-3,583\\ 
        & $r_b$ (\%)  & 0.0 & 0.0 & 0.0 & 25.68 & 7.18 & 0.0  & 0.0 & 67.15 & 100.0 \\ \hline
        \multirow{2}{*}{$l_{10}$ 2300 MW} & $\Delta \mathcal{U}_b$ (\$M)  & 181 & -2,333 & -10 & -149 & -1,439 & -146 & -910 & 2,288 & -2519\\ 
        & $r_b$ (\%) & 7.33 & 0.0 & 0.0 & 0.0 & 0.0 & 0.0 & 0.0 & 92.67 & 100.0  \\ \hline
        \multirow{2}{*}{$l_{12}$ 2300 MW} & $\Delta \mathcal{U}_b$ (\$M)  & 12 & 196 & -5 & 77 & 1,195 & 12 & -1848 & 336 & -25\\
        & $r_b$ (\%)  & 0.66 & 10.72 & 0.0 & 4.21 & 65.37 & 0.65 & 0.0 & 18.38 &100.0 \\ \hline 
        \multirow{2}{*}{Projects Sum} & $\Delta \mathcal{U}_b$ (\$M)  & 8,242 & -5,018 & -2,499 & -563 & -3,170 & -421 & -9,093 & 1,792 & -10,730\\
        & $r_b$ (\%)  & 40.54 & 13.57 & 0.0 & 5.11 & 10.39 & 0.09 & 0.0 & 29.69 &100.0 \\ \hline
        \multirow{2}{*}{Portfolio} & $\Delta \mathcal{U}_b$ (\$M) &6,215 &2,281 & -1,379 & -453 & -120 & 15 & -3,250 & 2,360 &5,669 \\
        & $r_b$ (\%) &57.17 &20.98 &0.0  & 0.0 & 0.0 & 0.14 & 0.0 & 21.71 &100.0\\ \hline
    \end{tabular}
\end{table}

Figure~\ref{fig:benefit_cost_ratio} shows the benefit\textendash cost ratio for the portfolio for each bus under each allocation method, with the benefits used for both subplots taken from the ``Portfolio'' row in Table~\ref{tab:project_allocation}. The result of the project-by-project allocation is that some loads, namely, those at $b_4$ and $b_5$, can be assigned positive cost despite having negative benefits from the overall portfolio. These positive allocations result from the positive benefits found for expansion on $l_6$, $l_7$, and $l_{12}$ when assessed individually and imply that the benefit\textendash cost ratio drops below zero for $b_4$ and $b_5$ in the project-by-project allocation. By contrast, the benefit\textendash cost ratio is consistent across all load zones in the case of the portfolio-level allocation, with negatively impacted zones assigned no cost. 
If subjected to the project-by-project allocation, loads at $b_4$ and $b_5$ would likely object and work to prevent the socially beneficial expansion from occurring, e.g., by denying necessary permits. To avoid this issue, we suggest that it is preferable to allocate costs for a portfolio rather than individual projects. 

\begin{figure}[t!]
\centering
\includegraphics[width=.6\linewidth]{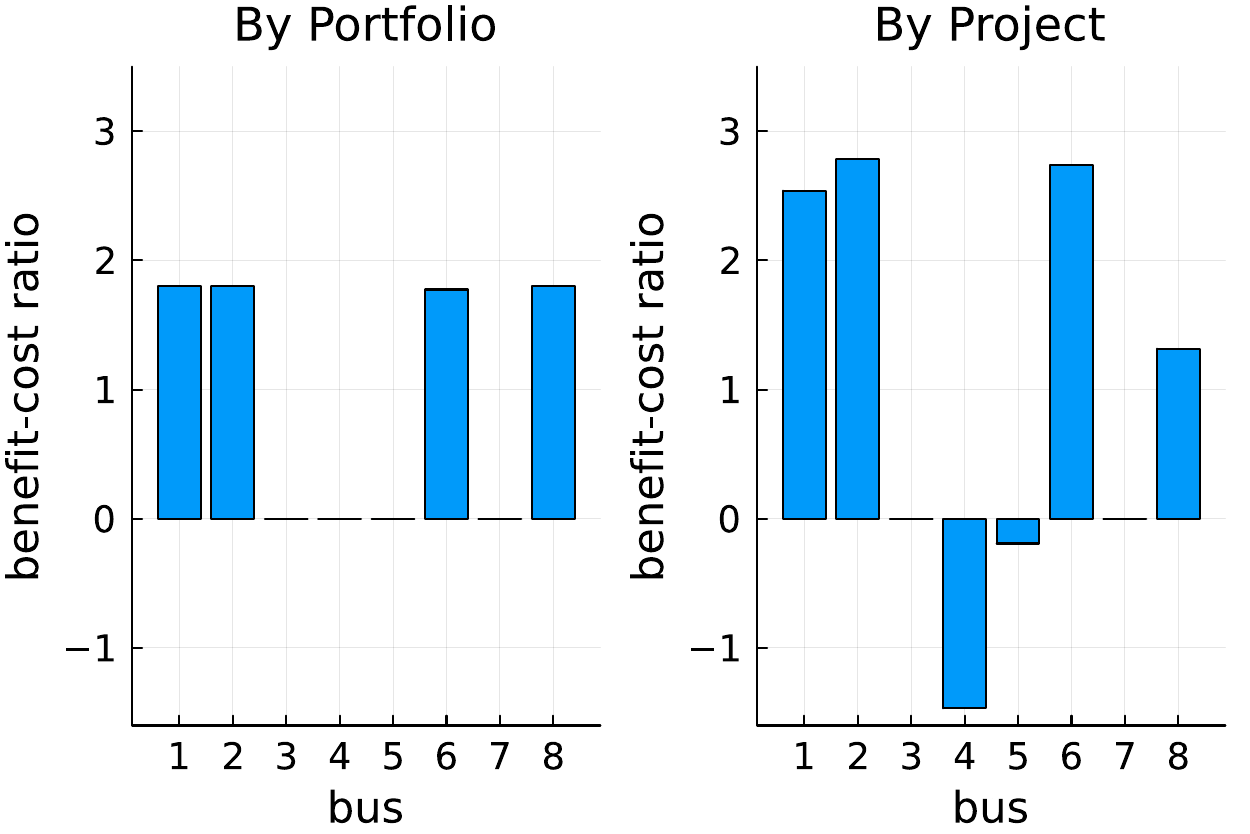}
  \captionof{figure}{Benefit\textendash cost ratio for different loads. While the benefit\textendash cost ratio is consistent across all buses when cost allocation is determined on a portfolio basis, the project-by-project allocation can lead to cost being allocated to loads that do not benefit from the overall portfolio.}
  \label{fig:benefit_cost_ratio}
\end{figure}

\subsection{Generator Impacts}

While the previous subsection considered the impacts on load only, we now turn to impacts on generation. A summary of the aggregate allocation across all generators and loads, calculated with Eqs.~\eqref{eq:load_allocation} and~\eqref{eq:gen_allocation}, is shown in Table~\ref{tab:project_allocation_w_gen}. We note that the aggregate benefits are substantially larger for generators than for loads in our case study, but we cannot make a general claim regarding how benefits are likely to be split in other instances. Consistent with intuition, we observe that the largest line expansion selected by the model, $l_2$, connects the zone in which the largest benefits accrue to generators, $b_3$, with the zone in which the largest benefits accrue to loads, $b_1$.

\begin{table}[!ht]
    \centering 
    \caption{Expected load and generation nodal benefits and cost allocation ratios for transmission expansion portfolio when allocating expansion cost to both load and existing generation. Generation benefit $\Delta \mathcal{U}^{gen}_b$ is the aggregation of positive benefits of all generation at bus $b$. Sum is the sum of positive benefits across buses.}
    \label{tab:project_allocation_w_gen}
        \begin{tabular}{ll|cccccccccc}
    \hline
        Participants && $b_1$ & $b_2$ & $b_3$ & $b_4$ & $b_5$ & $b_6$ & $b_7$ & $b_8$  & Sum\\ \hline
        load
        &$r^{load}_b$ (\%) &12.99 &4.77 &0.0  & 0.0 & 0.0 & 0.03 & 0.0 & 4.93 &22.72\\\hline 
       generation &$r^{gen}_b$ (\%)  & 0.0 & 0.01 & 46.48 & 11.2 & 0.03 & 0.07 & 17.3 & 2.19 & 77.28 \\ \hline
    \end{tabular}
\end{table}

As discussed in Section~\ref{sec:gen-benefits}, generators can also experience significant losses from expansion. The total expected benefits that accrue to existing generation of different types across buses is shown in Figure~\ref{fig:portfolio_gen_benefit}. Whereas the allocation in Table~\ref{tab:project_allocation_w_gen} aggregates only the positive benefits, Figure~\ref{fig:portfolio_gen_benefit} includes the negative impacts. Just as loads at $b_1$ and $b_2$ see the largest benefits from expansion, generators in those zones see the largest losses. The largest generator benefits occur at $b_3$, concentrated in existing thermal generators at that location.

\begin{figure}[t!]
\centering
  \centering
  \includegraphics[width=.85\linewidth]{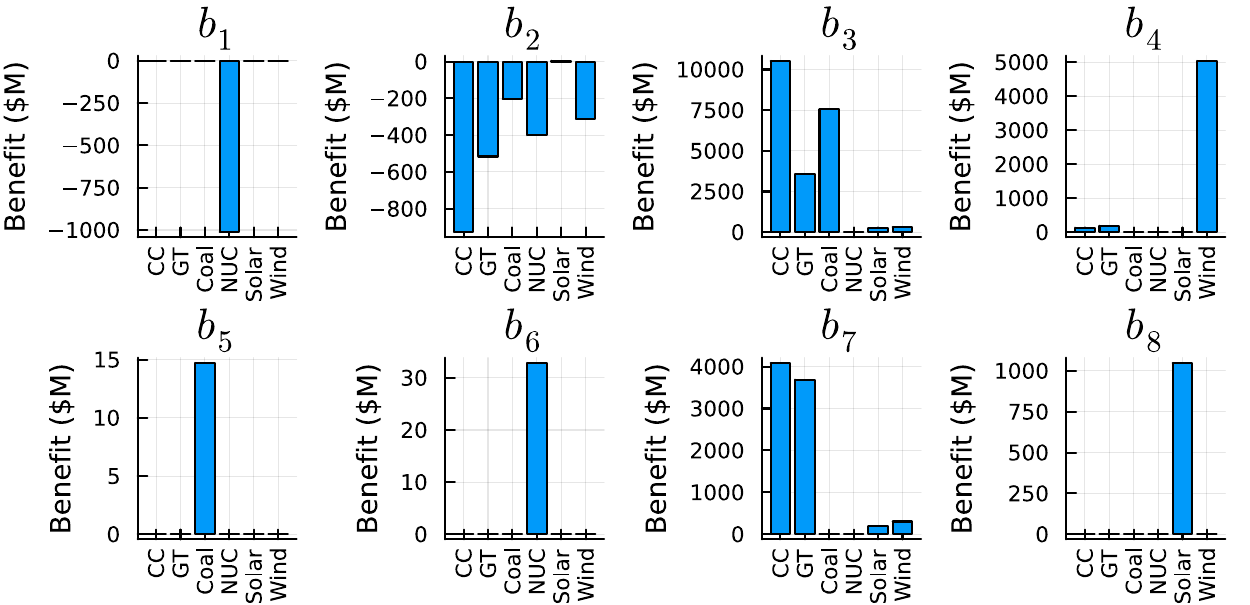}
  \captionof{figure}{Expected benefits of existing generation of different generation technologies across buses derived from portfolio.}
  \label{fig:portfolio_gen_benefit}
\end{figure}

\subsection{In-Sample vs. Out-of-Sample Tests}
The cost allocations defined above are calculated based on in-sample results, i.e., the expected zonal benefits determined as the weighted average across various scenarios where scenarios and scenario probabilities are taken from the planning model covering the whole planning horizon of 20 years. Table~\ref{tab:project_allocation} and~\ref{tab:project_allocation_w_gen} show the in-sample cost allocation ratios under two policies. As discussed above, however, the scenarios and probabilities determined for the planning model do not reflect the full range of possible outcomes or participant beliefs. Accordingly, a key question is the validity of these estimates and the extent to which out-of-sample results might diverge from the in-sample expected value. 

Estimating the benefits out of sample for the whole 20-year horizon is complicated computationally, because it would require definition of a complete policy describing how transmission and generation investments after the first stage will be made based on realizations of uncertainty that are not contained in our original planning model. In our context, such a policy cannot be defined: we rely on a stakeholder process to determine the scenarios and probabilities to be used in our planning model, and cannot fully specify the outcomes of future stakeholder processes. To avoid this issue, we instead perform out-of-sample tests for both cost allocation ratios on a single operating year. Specifically, we perform an out-of-sample analysis for $y=2$, i.e., year 2028, to assess benefits of transmission expansion projects determined in $y=1$, computing the distribution of realized benefits against the ex ante allocation. In out-of-sample tests, we used load and renewable availability data from $2022$, sourced and processed in a manner consistent with the procedures outlined in Section 4.1. Since our out-of-sample tests do not have transmission investment, we employ a linear program covering the entire year of data instead of selecting representative days as in the MIP planning model. Benefits are computed for all $3^5=243$ possible realizations of uncertainty described above. The aggregated generation and load benefits on each bus across different scenarios in year 2028 are shown in Figure~\ref{fig:oos_portfolio_gen_load_benefit}. The scenarios are ordered by the gross social benefits from the transmission expansion. It is noteworthy that while not generalizable, in our case study the expansion is beneficial under all realizations of uncertainty. It can also be observed that, while the rank ordering of zones is relatively stable overall, there are wide swings in the absolute benefits realized in each zone.

\begin{figure}[t!]
\centering
\begin{subfigure}[t]{.49\textwidth}
  \centering
  \includegraphics[width=.95\linewidth]{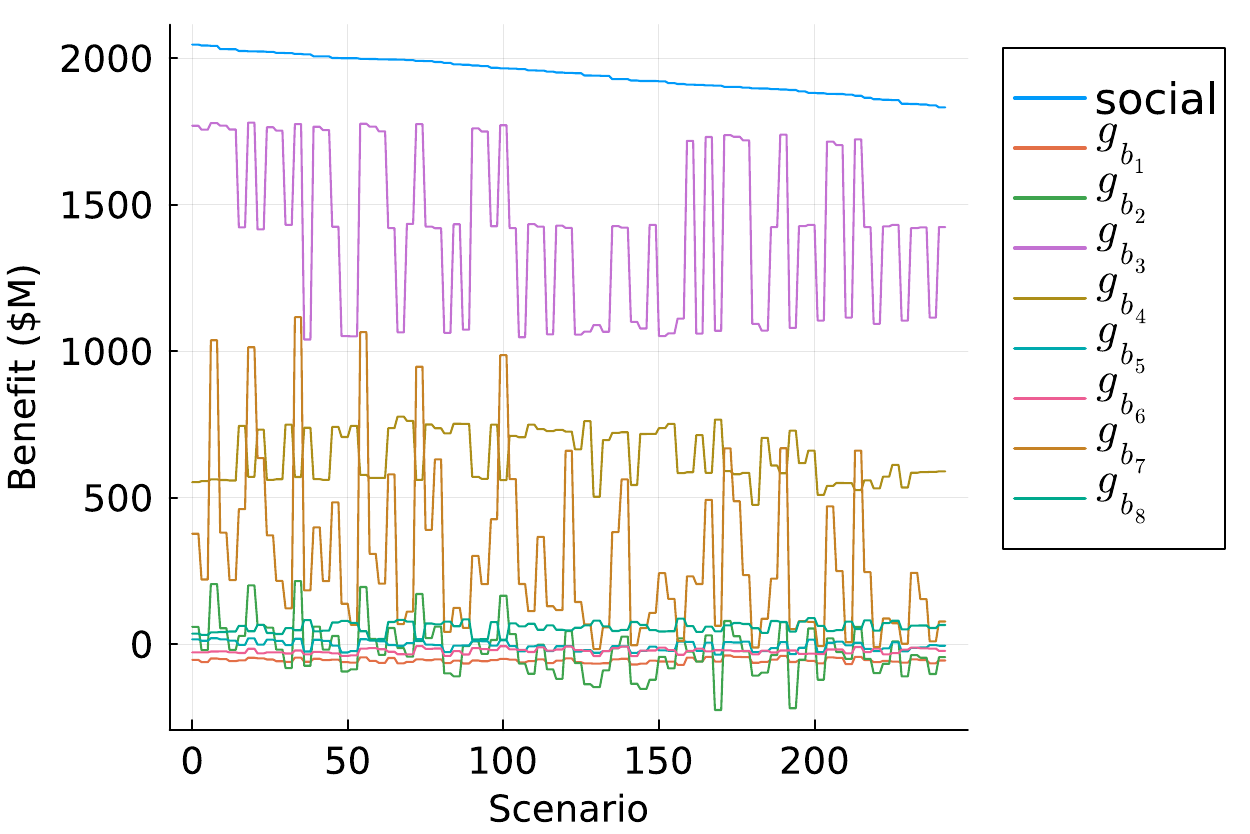}
  \caption{Nodal existing generation benefits.}
  \label{fig:oos_portfolio_gen_benefit}
\end{subfigure}
\begin{subfigure}[t]{.49\textwidth}
  \centering
  \includegraphics[width=.95\linewidth]{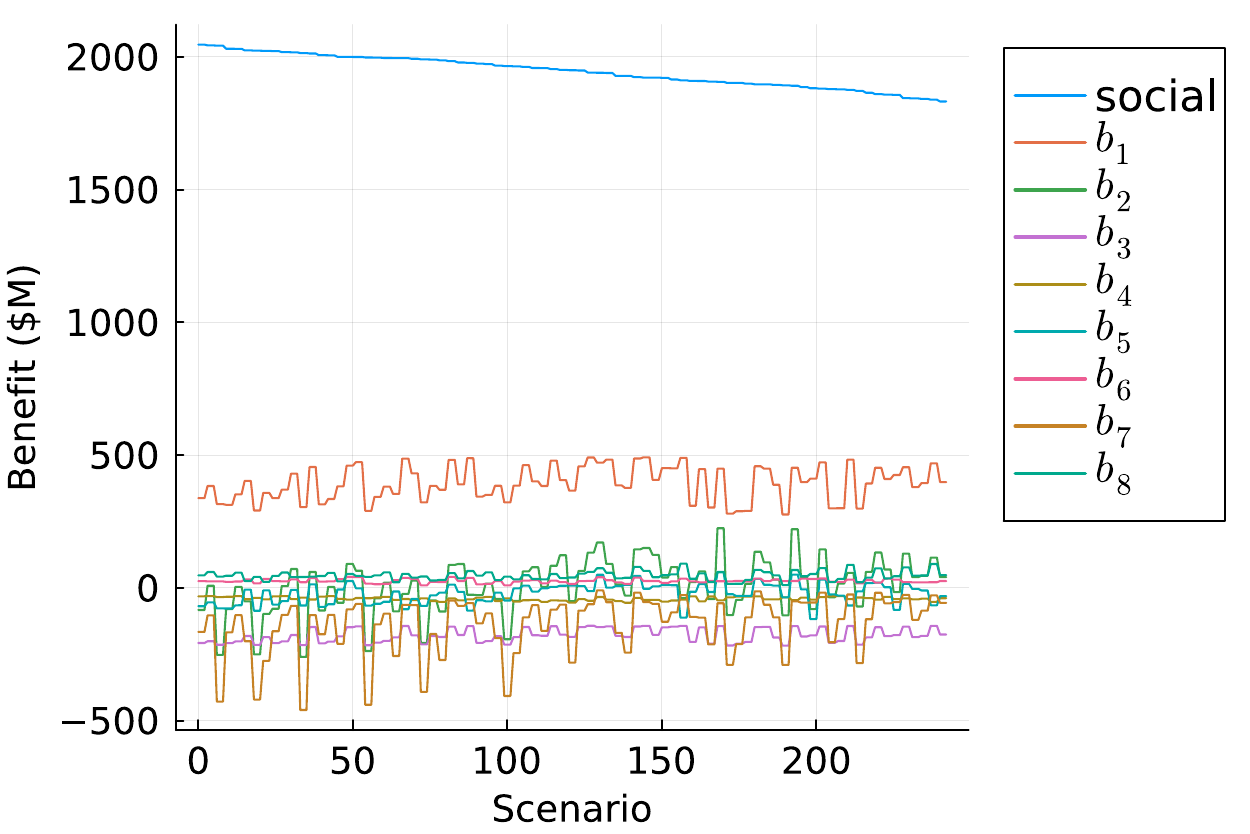}
  \caption{Nodal load benefits.}
  \label{fig:oos_portfolio_load_benefit}
\end{subfigure}
\caption{Benefits of the portfolio on out-of-sample tests in year 2028 ranked by social benefits.}
\label{fig:oos_portfolio_gen_load_benefit}
\end{figure}

The overall distribution of benefits evaluated ex post for generators and loads is shown in Figure~\ref{fig:oos_portfolio_ratio_w_gen}, with generation of all types aggregated at each bus. The blue bars indicate the number of times (out of the 243 scenarios) ex post benefits are calculated to be in each range, while the red dashed line indicates the cost allocation determined ex ante (as in Table~\ref{tab:project_allocation_w_gen}). Here, the consequences of uncertainty are apparent, as the ex post distribution of benefits in some cases does not contain the red dashed line. In absolute terms, the deviation can be quite significant: for example, generators at $b_3$ may see almost 70\% of total benefits from the portfolio after being allocated 46\% of costs. One possible reason for biased estimates is that ex ante benefits are estimated over a longer horizon than those calculated ex post. Even if a less biased ex ante estimate could have been produced with a more targeted computation, however, the significant variance observed in realized benefits would remain.

\begin{figure}[t!]
    \centering 
    \begin{subfigure}[t]{0.49\textwidth}
        \centering
        \includegraphics[height=2.1in]{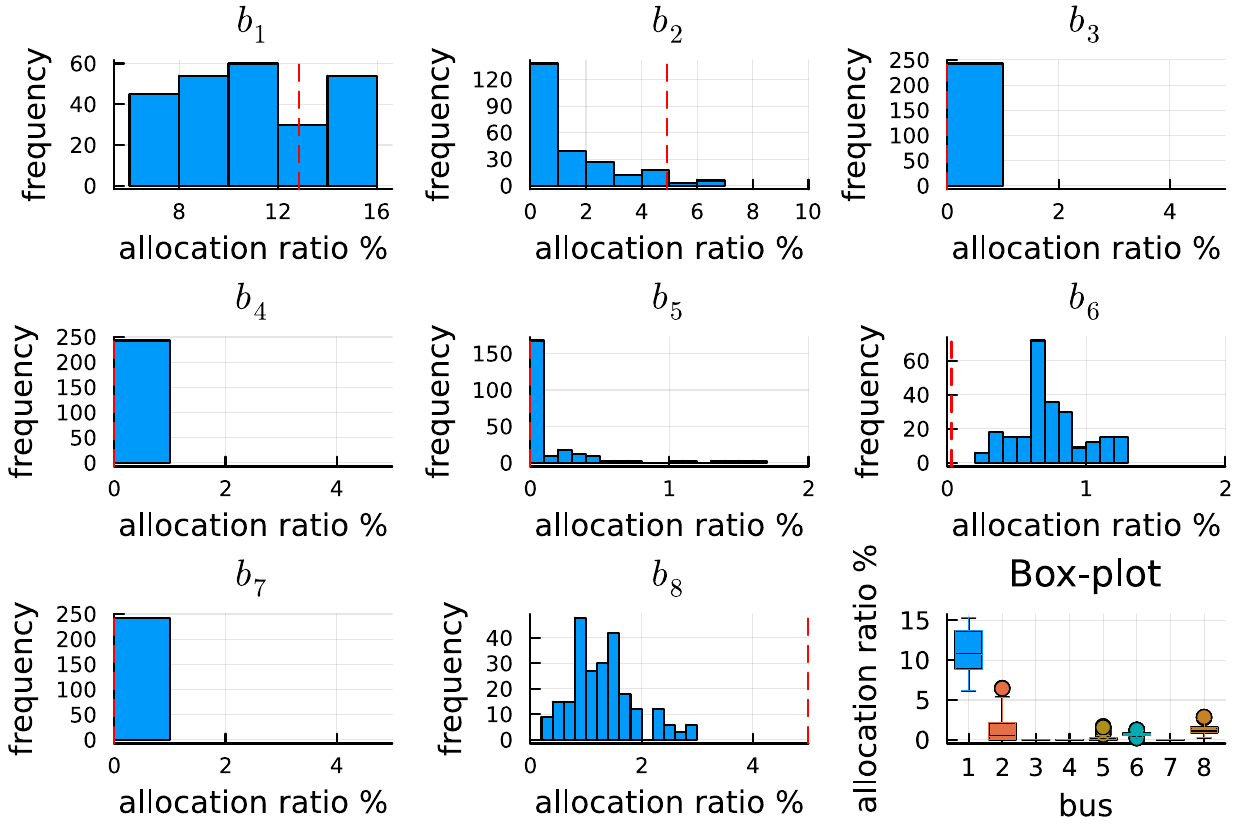}
        \caption{Distributional allocation ratio on load.}
        \label{fig:portfolio_load_ratio}
    \end{subfigure}
    \begin{subfigure}[t]{0.49\textwidth}
        \centering
        \includegraphics[height=2.1in]{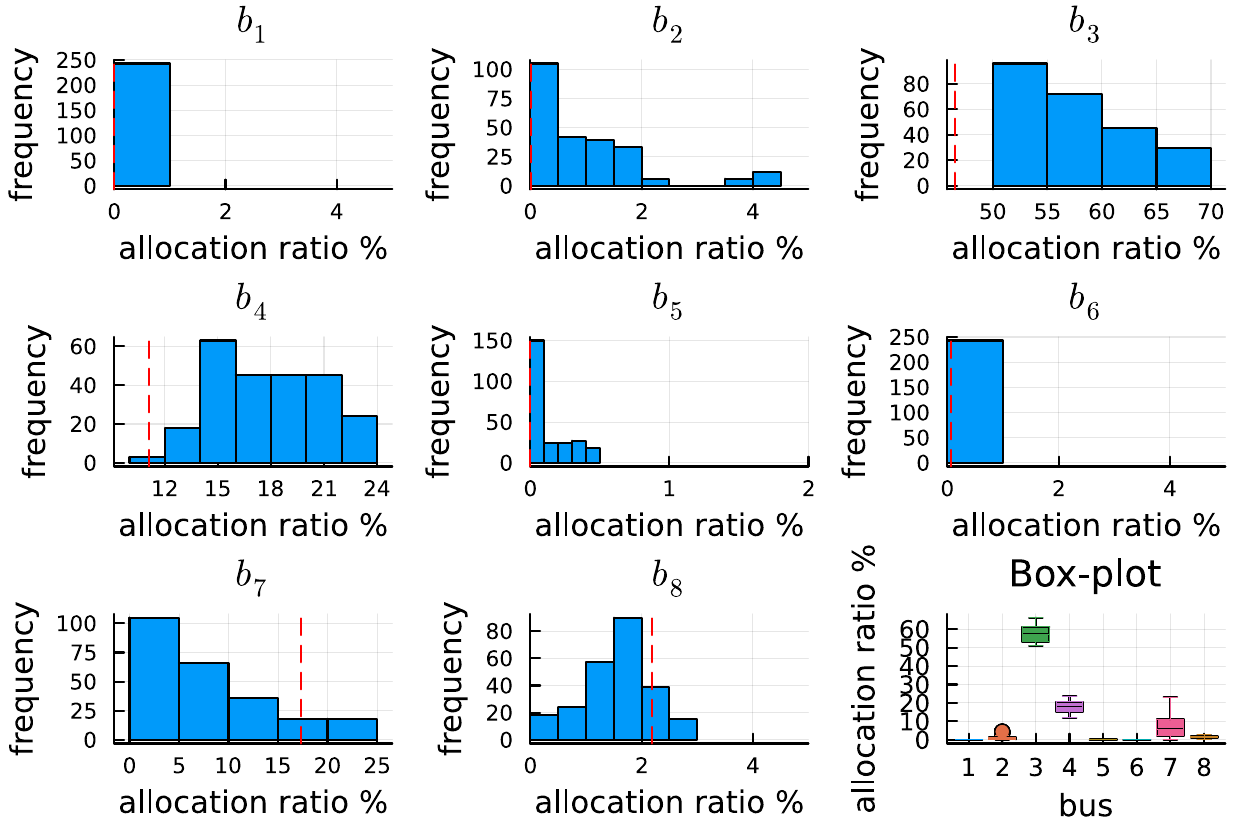}
        \caption{Distributional allocation ratio on generators.}
        \label{fig:portfolio_gen_ratio}
    \end{subfigure}
    \caption{Distributional allocation ratio of the transmission expansion portfolio for all buses on out-of-sample scenarios in year $2028$. Red dashed line is the portfolio allocation ratio in Table~\ref{tab:project_allocation_w_gen}.}
    \label{fig:oos_portfolio_ratio_w_gen}
\end{figure}

To test how the benefits of the portfolio may evolve over the life of the new lines, we construct an additional test using realizations of uncertainty for the year 2038 and assuming that a 3000 MW capacity expansion on $l_8$ has subsequently been added to the system. Referring back to Table~\ref{tab:add_line_inv}, an expansion on $l_8$ is chosen in each of the seven scenarios in the planning model, with the timing and size varying by scenario. With this 3000~MW expansion added to the system, we re-compute the benefits of the original portfolio of six transmission lines. Figure~\ref{fig:oos_portfolio_ratio} shows the distribution of benefits when allocated only to load for the years 2028 and 2038. In year 2038, benefits shift away from loads at $b_1$ to those at $b_2$, $b_5$, and $b_8$. 
\begin{figure}[t!]
    \centering 
    \begin{subfigure}[t]{0.49\textwidth}
        \centering
        \includegraphics[height=2.1in]{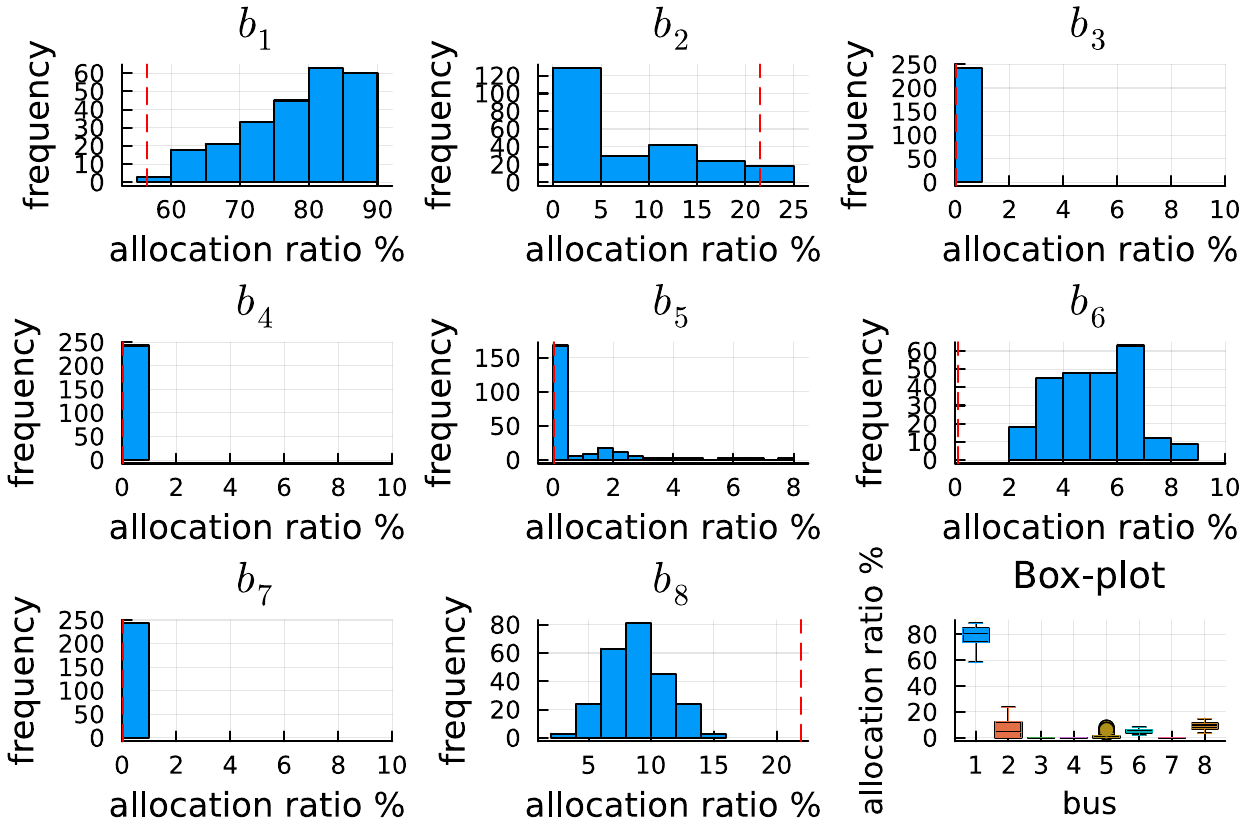}
        \caption{Distributional allocation ratio in year 2028.}
        \label{fig:portfolio_ratio_2028}
    \end{subfigure}
    \begin{subfigure}[t]{0.49\textwidth}
        \centering
        \includegraphics[height=2.1in]{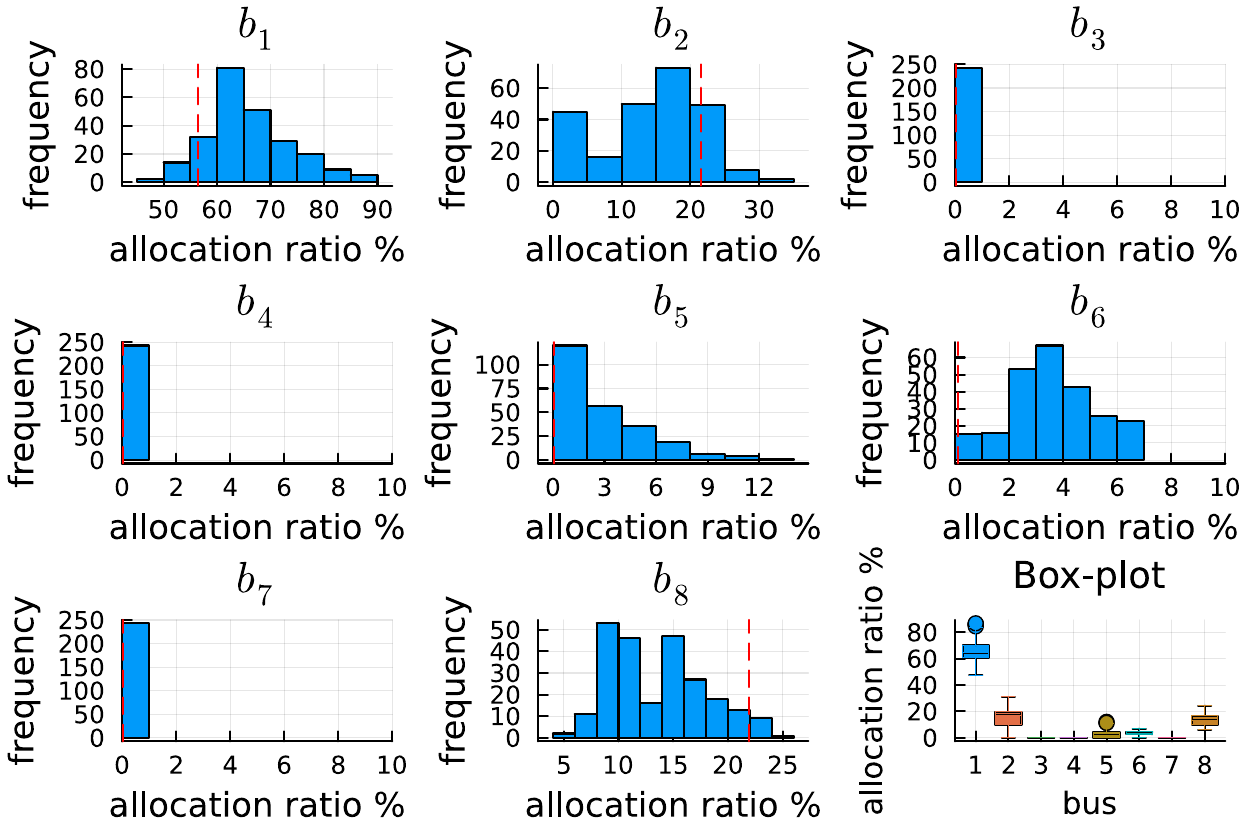}
        \caption{Distributional allocation ratio in year 2038 with an additional 3000 MW expansion on $l_8$.}
        \label{fig:portfolio_ratio_2038}
    \end{subfigure}
    \caption{Distributional allocation ratio of the transmission expansion portfolio for all buses and generation on out-of-sample scenarios. Red dashed line is the portfolio allocation ratio in Table~\ref{tab:project_allocation}.}
    \label{fig:oos_portfolio_ratio}
\end{figure}
The shift exhibited in Figure~\ref{fig:oos_portfolio_ratio} suggests an extension of the argument in Section~\ref{se:port-proj} that a portfolio-level allocation should be preferred to a project-level allocation. Suppose two projects with 50-year expected lives are selected and built in consecutive years. Given that they will coexist in the network for 49 out of their 50 years, their benefits will necessarily be interdependent and could be better assessed jointly. Extending the argument further, estimates of the aggregate benefits of the network may be more accurate than estimates of the benefits provided by any subset of network elements.

Overall, the results confirm the potential for uncertainty to cause challenges in cost allocation given the disagreements that market participants will inevitably have on the probability of future scenarios. In the context of the ``beneficiaries pay'' standard, the distribution of possible outcomes makes it clear that an allocation of costs determined ex ante will not be commensurate with the benefits realized ex post. Economic theory offers a potential resolution to the resulting conflicts in the form of financial contracts issued ex ante that would effectively reallocate cost to the ultimate beneficiaries~\citep{Ferris2022}. Given the complications involved in defining such contracts, we defer the effort to future work.  

\section{Conclusion} \label{se:conclusion}
Given the numerous ways in which motivated parties can intervene to prevent transmission expansion, disputes over cost allocation can hold up investment in regional and interregional projects. Out of fairness and to forestall such interventions, U.S. system planners have sought methods to allocate costs according to the estimated benefits that projects will bring. In a direct benefits modeling approach, planners could in principle solve an optimization model that both established the social benefits of a project and enabled an estimate of benefits at the participant level. However, inadequacies in both the models available and the information used in them can lead to significant disagreements about the fairness of the resulting allocations.

This paper identifies several challenges in the use of models to establish cost allocations. Given the complexity of the modeling task, planners typically use a combination of software tools to evaluate proposed projects. One consequence is that it may be difficult to establish a valid counterfactual against which benefits can be measured and to calculate all categories of benefits that could result from an expansion project. The challenge is even greater when assessing benefits out of sample, since a full calculation would require not only determining the range of scenarios to be tested but also specifying a policy for future expansion decisions given the realization of uncertainty.

Without fully resolving these challenges, the theoretical analysis and numerical study lead to five observations connected to the ``beneficiaries pay'' principle. First, benefit estimates should include some attempt to account for the change in the resource mix that is likely to occur with any change in the network. Second, cost should in general not be allocated to new entrants, but should be allocated to incumbents that benefit from transmission expansion. Third, allocations made on the basis of portfolios of projects are likely to be more defensible than those made on individual projects. Fourth, conflicts might be lessened with greater effort to compensate the losers from socially beneficial transmission expansion. Fifth, conflicts might be lessened with greater effort to address the risk that participant-level benefits will diverge significantly from ex ante allocation decisions.

\section*{Acknowledgements}

This research was supported by the Power Systems Engineering Research Center under Project M-43. The authors would like to thank Jim McCalley, Gustavo Cuello Polo, Ali Jahanbani Ardakani, Richard O'Neill, Abe Silverman, and members of the project advisory team for feedback and discussions. 

 \newcommand{\noop}[1]{}

\end{document}